\documentclass[numbers,webpdf,imaiai]{ima-authoring-template}%

\graphicspath{{Fig/}}

\theoremstyle{thmstyletwo}%
\newtheorem{theorem}{Theorem}%  meant for continuous numbers
\newtheorem{proposition}{Proposition}% to get separate numbers for theorem and proposition etc.

\newtheorem{remark}{Remark}%
\newtheorem{lemma}{Lemma}%

\newcommand{\Z}{\mathbb{Z}}
\newcommand{\R}{\mathbb{R}}
\newcommand{\C}{\mathbb{C}}

\newcommand{\F}{\mathcal{F}}
\newcommand{\B}{\mathcal{B}}
\newcommand{\E}{\mathrm{E}}
\newcommand{\A}{\mathcal{A}}
\newcommand{\I}{{\mathrm{I}}}
\renewcommand{\i}{\mathrm{i}}
\renewcommand{\P}{\mathbb{P}}
\newcommand{\Q}{\pmb{Q}}
\newcommand{\ID}{\pmb{I}}
\newcommand{\p}{\mathtt{p}}
\newcommand{\w}{\omega}
\newcommand{\e}{e}
\newcommand{\eps}{\varepsilon}
\newcommand{\mJ}{\mathcal{J}}
\newcommand{\mX}{\mathcal{X}}

\newcommand{\mG}{\mathcal{G}}
\newcommand{\mQ}{\mathcal{Q}}

\newcommand{\sign}{\operatorname{sign}}
\newcommand{\argmin}{\operatorname{arg\;min}}
\newcommand{\supp}{\operatorname{supp}}
\newcommand{\Arg}{\operatorname{Arg}}
\newcommand{\AP}{\operatorname{AP}}
\renewcommand{\Re}{\operatorname{Re}}
\renewcommand{\Im}{\operatorname{Im}}

\numberwithin{equation}{section}

\begin{document}

\DOI{DOI HERE}
\copyrightyear{}
\vol{}
\pubyear{}
% \access{Advance Access Publication Date: Day Month Year}
% \appnotes{Paper}
% \copyrightstatement{Published by Oxford University Press on behalf of the Institute of Mathematics and its Applications. All rights reserved.}
% \firstpage{1}

%\subtitle{Subject Section}

\title[Statistical inference based on band-limited kernels]{Statistical inference based on band-limited kernels: \\ 
Rational-infinitely divisible distributions and beyond}

\author{Vladimir Panov*\ORCID{0000-0001-8395-1909}
\address{ }
}
\author{Anton Ryabchenko %\ORCID{0000-0000-0000-0000}
\address{\orgdiv{Laboratory of Stochastic Analysis and its Applications}, \orgname{HSE University}, \orgaddress{\street{Pokrovsky boulevard 11}, \postcode{109028}, \state{Moscow}, \country{Russia}}}} 

\authormark{V. Panov and A. Ryabchenko}

\corresp[*]{\href{vpanov@hse.ru}{vpanov@hse.ru}}

\received{Date}{0}{Year}
\revised{Date}{0}{Year}
\accepted{Date}{0}{Year}

%\editor{Associate Editor: Name}

\abstract{
This paper investigates the problem of statistical inference for a mixture distribution consisting of a discrete and a continuous component, with a particular focus on the class \(\Q\) of rational-infinitely divisible  distributions. We consider non-parametric estimation of both components of the mixture as well as  the quasi-L{\'e}vy measure, assuming that the mixture belongs to the class \(\Q\). We propose an estimation framework based on band-limited kernels, which are the functions characterized by compactly supported Fourier transform. Under mild assumptions, the proposed estimators are theoretically shown to achieve polynomial (and in some cases even almost parametric) convergence rates. Finally, we demonstrate
the numerical performance of the algorithm on simulated examples.}

% \boxedtext{
% \begin{itemize}
% \item Key boxed text here.
% \item Key boxed text here.
% \item Key boxed text here.
% \end{itemize}}

\maketitle
\thispagestyle{empty}

\section{Introduction}

The class $\Q$ of rational-infinitely divisible distributions (also known as quasi-infinitely divisible distributions) was introduced recently in the paper by Lindner, Pan, and Sato \citep{LPS2018}, and is defined as follows. 
%introduced this class of distributions and established first theoretical results. 
A probability measure $\mu$ belongs to $\Q$ if there exist two distributions $\mu_1$, $\mu_2$ from the class $\ID$ of infinitely divisible distributions such that 
\begin{equation*}
    \mu_1 = \mu * \mu_2,
\end{equation*}
where $*$ denotes the convolution operator.  The term "rational-infinitely divisible distribution" was coined by Khartov \citep{Khartov2026} and outlines that the characteristic function of $\mu$ is  actually a ratio between the characteristic functions of two infinitely divisible distributions.

All distributions from $\ID$ belong to $\Q$ as well, but the converse is far from true: the class $\Q$ is significantly larger. In particular, it includes the Bernoulli distribution with parameter not equal to $1/2$, as well as normal mixtures with centered components. Almost all known examples of the distributions from $\Q \setminus \ID $ that have been described in the literature,  belong to the subclasses, which we will discuss below. 
\begin{enumerate}
    \item Distributions with a non-zero discrete component, i.e.
    \begin{equation} \label{first_model}
        \mu = \w \mu_d + (1 - \w) \mu_c,
    \end{equation}
where $\mu_d$ and $\mu_c$ are, respectively, discrete and continuous distributions, and $\w \in (0, 1]$. A fundamental result in this field states that  $\mu \in \Q$   if and only if its characteristic function is bounded away from zero.  This fact was first proven for pure discrete distributions ($\w=1$) by Alexeev and Khartov \citep{AK2023}, later extended to the case when $\w \in (0,1]$ and $\mu_c$ is  absolutely continuous  (Berger and Kutlu \citep{BK2023}), and finally generalised to all distributions of the form \eqref{first_model}, including also singular continuous $\mu_c$ (Khartov \citep{Khartov2026}). 
 An interesting observation is that $\mu \in \Q$  if and only if  $\mu_d \in \Q$  and   the characteristic function of $\mu$ doesn't have real zeros (see Khartov \citep{Khartov2026}, Theorem~3). At the same time, the continuous component $\mu_c$ of a distribution \(\mu \in \Q\) may be either a rational-infinitely divisible distribution or it may not be.
     \item 
    Contamination model
        \begin{equation}\label{second_model}
        \mu = \w \mu_m + (1 - \w)\mu_e,
    \end{equation}
    where $\w \in (1/2, 1)$, 
 $\mu_m$ (``main component'') belongs to $\Q$, and $\mu_e$ (``error'') can be any distribution.  Lindner, Pan and Sato \citep{LPS2018} showed that \(\mu \in \Q\) if  the ratio between the characteristic functions of \(\mu_e\) and \(\mu_m\) is the Fourier transform of a measure with total variation less than \(\w/(1-\w)\).  The main difference  between this model and the rational-infinitely divisible distribution in the form \eqref{first_model}  is that \(\mu_m\) may not be discrete, but this generalization comes with the cost of requiring  \(\w>1/2\) and the additional assumption on the characteristic functions presented above.  It is interesting to note that a subclass of distributions of type (\ref{second_model}) that  have an atom of mass larger or equal to $1/2$, was shown to be rational-infinitely divisible nearly $50$ years before the extensive study of the class~\(\Q\) began (Cuppens \citep{Cuppens1970}).
\end{enumerate}
Rational-infinitely divisible distributions represent an innovative tool that can be applied across different fields. In particular, they have been used to show that a Cram{\'e}r\,--\,Wold device for infinite divisibility of \(\Z^d\)-valued distributions holds (Berger and Lindner~\citep{BL2022}). Also, they have been employed in various fields such as financial modeling (Madan et al. \citep{Madan2023}), physics (Demni and Mouayn \citep{DM2015}), number theory (Nakamura~\citep{Nakamura2022}) and insurance mathematics (Zhang et al.~\citep{Zhang2014}). 

As usual, the wide range of applications motivates the development of  statistical methods. 
While statistical estimation for the class of infinitely divisible distributions and L{\'e}vy-based models has been covered by a large number of studies, estimation in class \(\Q\) has only been investigated in a few papers. The first one is Passeggeri \cite{Pass2023}, where the estimation was considered in the framework of Bayesian analysis. Recently, Panov and Ryabchenko \cite{PanovRyabchenko2026} proposed a Fourier  approach for distributions of the form (\ref{second_model}). Their method is based on an analogue of the L{\'e}vy\,--\,Khintchine formula for the class \(\Q\), which differs from the classical version of this formula only in the use of a signed L{\'e}vy measure, also known as quasi-L{\'e}vy measure. 
However, there seems to be a lack of research on statistical inference for distributions of the type~(\ref{first_model}) from the class \(\Q\). Our paper intends to address this issue.

\subsection{Contribution}

Our research focuses on statistical inference for  distributions of the type (\ref{first_model}).  
We propose an estimation approach based on the band-limited kernels, which are defined as functions with Fourier transform having compact support. The use of these functions, in particular the $\mathrm{sinc}$ function, has been widely studied in signal processing, starting with the pioneering work of Whittaker \cite{Whittaker1915}. However, their applications in statistics are not well known.  In this paper, we present two different approaches based on band-limited kernels for solving problems related to distributions of the form (\ref{first_model}), where  $\mu_c$ is an absolutely continuous distribution:

%We propose an estimation approach based on band-limited kernels. These kernels are defined as functions with a compact Fourier transform, and they are widely used in signal processing. However, their statistical application has not been well-studied

% One of the most commonly used band-limited kernels are based on the function $\mathrm{sinc}(x)$. 
 
\begin{enumerate}
\item  estimation of $\w$, $\mu_d$ and $\mu_c$;% for the observations of $\mu$ using an approach based on convolution of characteristic function $\phi$ of $\mu$ with band-limited kernels;
\item  estimation of the quasi-L{\'e}vy measure $\nu$ of $\mu$ and quasi-L{\'e}vy measure $\nu_d$ of $\mu_d$ provided that \(\mu \in \Q\). %This algorithm involves convolutions with band-limited kernels as well. These estimates also allow to estimate the part of $\nu$ that corresponds to $\mu_2$ (we denote it as $\nu_2)$. This problem was studied in the ID case in Belomestny, Reiss~\cite{BR2015}, but, to the best of our knowledge, it is new in the context of QID distributions, which are not ID.
\end{enumerate}
It is a worth noting that the solution to the second problem is based on the estimate of \(\mu_d,\) which is obtained by solving the first problem.

We propose an approach for the estimation of $\w$, $\mu_d$ and $\mu_c$%address the first problem 
using concepts from the theory
of almost periodic functions \(\AP(\R, \C)\) %and Besicovitch spaces 
(see Corduneanu \cite{Cord2009}), in particular, convolution of these functions with band-limited kernels. This methodology gives rise to statistical inference in the class \(\AP(\R, \C)\) itself. 
%and show the solutions of these pr based on the for the case when atoms of $\mu_d$ are separated from each other and with this assumption we are able to show that the resulting estimates have polynomial rates of convergence.  
% This fact is an interesting observation when compared to the fact that the convergence rates of similar methods for ID distributions are logarithmic and can't be improved in general, see Theorems~4.3, 4.4, 5.7 in Belomestny and Reiss~\cite{BR2015}. 
To estimate \(\nu\), we suggest
an inverse Fourier-based method that also utilizes band-limited kernels. This approach
was previously applied by Belomestny and Reiss \cite{BR2015} for infinitely divisible distributions, but, to our knowledge,
its application to the class \(\Q\) has never been considered before. We demonstrate that both estimators achieve polynomial rates of convergence
given some assumptions on \(\mu\). 

\subsection{Structure}
The paper is organised as follows. The next section offers an overview of almost periodic functions and rational-infinitely divisible distributions.  Next, in Sections \ref{sec_stat_approach} and \ref{secqid} we separately consider two statistical problems mentioned above.  Section \ref{numerical_results} deals with the numerical study of the considered methods. The proofs are collected in Section \ref{proofs}. %Some auxiliarly facts are given in Appendix \ref{service_appendix}.

\section{Preliminaries}
\subsection{Almost periodic functions}
%Throughout the paper we consider almost periodic functions in Bohr's sense, this class is constructed in the following way.
In Bohr's sense, the family of almost periodic functions  $\AP(\R, \C)$  is defined as the closure of the set of functions 
 \begin{equation*}
\Bigl\{ 
 f: \R \to \C \;\; |
 \;\;
 f(t)=\sum_{k=1}^{n} a_k e^{i\lambda_kt}, \;\; a_k \in \C, \; \lambda_k \in \R, \; k=1, 2, ..., n
 \Bigr\},
 \end{equation*}
in the sup-norm \(
||f|| = \sup(|f(t)|, \; t \in \R)\). The term ``almost periodic function''  comes from an equivalent definition: $f \in \AP(\R, \C)$ if and only if for every $\varepsilon > 0$ there exists a positive number $l = l(\varepsilon)$ such that any interval $(a, a + l) \subset \R$ contains a number $\tau$ with the property
\begin{equation*}
    |f(t + \tau) - f(t)| < \varepsilon, \qquad t \in \R.
\end{equation*}
We refer to  \cite{Cord2009} for a proof of the equivalence of these definitions and a comprehensive study of this class. In our study we will use the following result, which is given as Proposition~3.8 in \cite{Cord2009}.

\begin{lemma}\label{lem1}
The convolution of an almost periodic function 
$f \in \AP(\R, \C)$ and any function $K \in L^1(\R, \C)$ is also an almost periodic function, that is,  
\begin{equation*}
    (K * f)(t) = \int_{\R} K(s)f(t - s)ds \in \AP(\R, \C).
\end{equation*}
\end{lemma}

Note that the characteristic function of any discrete distribution belongs to $\AP(\R, \C)$. The proof of Lemma \ref{lem1} for this case is fairly straightforward: the   convolution of a kernel $K \in L^1(\R, \C)$ and  the characteristic function $\phi_d(u)$ of a discrete distribution  with atoms at points $x_k$ and  corresponding probabilities $p_k, \, k = 1,2,...$, results in
\begin{align} \label{line1} 
     (K * \phi_d)(u) =  \int_{\R} K(v)\Bigl[\sum_{k=1}^{\infty} p_k e^{\i x_k (u-v)} \Bigr] dv = \sum_{k = 1}^{\infty}\underbrace{p_k \F[K](-x_k)}_{=:a_k}e^{iux_k} \in \AP(\R, \C),
\end{align} 
where $\F[K](\cdot)$ is the Fourier transform of the kernel \(K\). 

\subsection{Rational-infinitely divisible distributions} 
The L{\'e}vy\,--\,Khintchine formula states that the characteristic function $\phi(u)$ of an infinitely divisible distribution can be represented in the following form 
\begin{align} \label{LC}
    \phi(u)
    =
    \exp\Bigl\{
      \i \gamma u-\frac{1}{2}\sigma^2 u^2
			    +
		    \int_{\mathbb{R}\setminus \{0\}}
		    \left( 
		    		e^{\i u x}-1 - \i u x \I\{x \in [-1,1]\}		\right)			\nu(dx)\Bigr\},
\end{align}
where \(\gamma \in \R, \sigma \in \R_+,\) \(\nu: \B(\R \setminus \{0\}) \to \R_+\) is a measure such that  \(\int_{\R}  \min(1, x^2)\nu(dx) < \infty.\) 
%and $c: \mathbb{R} \rightarrow \mathbb{R}$ is a so-called representation function, which is a bounded, Borel measurable function, satisfying \(c(0)=0\) and
%satisfies \(
%c(x) = x+ o(x^2)\) as \(x \to 0\).
%Typically, the function $c(x)$ is chosen as $c(x) = x\I_{[-1, 1]}(x)$.

The characteristic function of the distributions from the class \(\Q\) can be represented in a similar way, with a difference only in the measure \(\nu\). For the distribution from class \(\Q\), this measure is signed, in the sense that it can take negative values. The total variation \(|\nu|\) of this measure satisfies the condition \[\int_{\R}  \min(1, x^2)|\nu|(dx) < \infty.\] 
As in the case of \(\ID\), the triplet \((\gamma, \sigma^2, \nu)\) completely describes the distribution from \(\Q\).
% Unfortunately, a simple characterization of QID class is not known yet, unlike for ID distributions, where it is a well-known fact that their characteristic functions are not zero in any point. However, some results were proven for certain subclasses of distributions. For example,  a distribution on integers  is QID if and only if its characteristic function doesn't have real zeros (Lindner, Pan and Sato~\cite{LPS2018}). More generally,  a discrete distribution is QID if and only if its characteristic function is separated from zero (Alexeev and Khartov~\cite{AK2023}). This result also holds for  mixtures of discrete and absolutely continuous distributions, which are considered in this paper (Berger and Kutlu~\cite{BK2023}). As of today the most researched and fully described subclass of QID distributions is the class of mixtures with at least one discrete component.

The following proposition summarizes some results from \cite{BK2023} and \cite{LPS2018}, and provides the exact form of the characteristic triplet for two types of rational-infinitely divisible distributions, which are described in the introduction.

\begin{proposition} \label{propmain}
(i)  Let $\mu\in \Q$ be a distribution of the type (\ref{first_model}), where $\mu_c$ is an absolutely continuous distribution. Then its triplet is equal to \((\gamma,0, \nu)\) with some \(\gamma>0\) and  quasi-L{\'e}vy measure \(\nu\) in the form 
\begin{eqnarray}\label{nu1}
\nu(dx) =  \sum_{y \in C} c_y \delta_y (dx)+ \bigl( 
h(x) + r \frac{e^{-|x|}}{|x|}\sign(x) 
\bigr) dx,
\end{eqnarray}
where the set $C \subset \R$, the sequence $(c_y)_{y \in C} \in \R$, $h \in L^1(\R, \C)$ and $r \in \Z$.

(ii) Let $\mu$ be a distribution  of the type (\ref{second_model}), where $\mu_m \in \Q$  with  triplet $(\gamma, \sigma^2, \nu)$, and $\mu_e$ is some distribution on \(\R.\) Let 
the ratio between the characteristic functions of \(\mu_e\) and \(\mu_m\) be the Fourier transform of a measure \(\Lambda\) with total variation less than \(\w/(1-\w)\). Then $\mu \in \Q$ with the characteristic triplet \[\Bigl(\gamma + \int_{-1}^1 x \Tilde{\nu}(\mathrm{d}x),  \quad \sigma^2, \quad  \nu + \Tilde{\nu}\Bigr),\]
where $\Tilde{\nu}$ is a finite signed measure defined as follows: \begin{equation}\label{Tildenu}
    \Tilde{\nu}(dx) = \sum_{k = 1}^{\infty} \frac{(-1)^{k + 1}}{k}\Bigl( 
\frac{1-\w}{\w} 
\Bigr)^k\Lambda^{*k}(dx),% \qquad B \in \B(\R\setminus\{0\}),
\end{equation}
with $\Lambda^{*k}$ being the convolution of the measure $\Lambda$ with itself  $k$ times.
\end{proposition}

\section{Inference for mixtures of type (\ref{first_model})}\label{sec_stat_approach}
Consider the mixture $\mu$ of the type (\ref{first_model}), where the second component is an absolutely continuous distribution, denoted by $\mu_{ac}$. Let the characteristic functions of the distributions \(\mu_d, \mu_{ac}, \mu\) be denoted by 
\begin{align} \label{phiw}
\phi_d (u) = \sum_{k = 1}^{\infty} p_k &e^{iux_k}, \qquad 
\phi_{ac} (u) = \int_{\R} g_{ac}(x) e^{iux} dx, \nonumber \\
\phi(u) &= \w \phi_d(u) + (1 - \w) \phi_{ac}(u),
\end{align}
where \(x_k \in \R,\) \(p_k >0,\) \(\sum_{k=1}^\infty p_k =1,\)  and  $g_{ac}(x)$ is the density function of $\mu_{ac}$.  In this section we aim to estimate both components, \(\mu_d\) and \(\mu_{ac}\), as well as  the parameter \(\w\) from the observations \(X_1,..., X_n\), which are drawn from the distribution \(\mu.\) 

\subsection{The ideas behind the estimation procedure}  
\begin{enumerate} 
\item Consider the convolution of the characteristic function $\phi$ and a function $K \in L^1(\R)$,
\begin{align}
    (K * \phi)(u) &= 
        \w (K * \phi_d)(u) +   (1-\w)  (K * \phi_{ac})(u).%\int_{\R} K(v) \phi(u - v) dv \\
%    &=& \w  \int_{\R} K(v) \Bigl( \sum_{k = 1}^{n} p_ke^{i(u - v)x_k} \Bigr) dv \\ 
  %  &&\hspace{2cm}+ (1- \w) \int_{\R} K(v)\int_{\R}g_{ac}(x)e^{i(u - v)x}dxdv.
\end{align}
%where the first summand can be considered following the same lines as in (\ref{eq1}) with the proper choice of the kernel. At the same time, 
The application of the Fubini theorem to the second term gives
\begin{align*}
    (K * \phi_{ac})(u) &=
     \int_{\R} K(v)\int_{\R}g_{ac}(x)e^{i(u - v)x}dx dv 
     =
\int_{\R} g_{ac}(x)\F[K](-x)e^{iux} dx.
\end{align*}
Note that  \(g_{ac}(\cdot)\F[K](\cdot) \in L^1(\R),\) since \(\F[K](\cdot)\) is bounded due to \(K \in L^1(\R)\). Therefore, 
\begin{equation*}
	(K * \phi_{ac})(u)  \to 0  \qquad  \text{as} \;\; |u|  \to \infty,
\end{equation*} 
by the Riemann\,--\,Lebesgue lemma.
%\begin{equation*}
%  %  \F[\F[K_{a,b}]g_{ac}](u) 
%  \int_{\R} g_{ac}(x)\F[K_{c,\delta}](-x)e^{iux} dx \rightarrow 0, \qquad u \to \infty,
%\end{equation*}
%due to the application of the Riemann\,--\,Lebesgue lemma to the function \(g_{ac}(x)\F[K_{c,\delta}](-x)  = g_{ac}(x) \I\bigl(x \in [c-\delta, c+\delta]\bigr) \in L^1(\R).\)
Thus, we get for large \(|u|\)
\begin{align}\label{approx}
    (K * \phi)(u) &\approx \w (K * \phi_d)(u).
\end{align}
This observation  gives rise to the estimation procedure, which consists of the  estimation of \(\mu_d\) (that is, the atoms \(x_k\) and the corresponding probabilities \(p_k\)), and further estimation of the density function \(g_{ac}.\) \newline 

%This section discusses statistical inference for mixtures of the form \eqref{first_model}. % Before we proceed to the first statistical task, let us briefly describe the key idea behind the estimation procedure.
%\subsection{Inference for $\AP(\R, \C)$}
%We arrive at the conclusion, that for \(\w>0\)
\item In what follows, we will consider the discrete distributions \(\mu_d\), which satisfy the following assumption. 
\begin{enumerate}
\item[(A1)] \label{cond_A1} \(\mu_d\) is supported on a discrete set \(\{x_1, x_2,...\}\) such that for some \(\delta>0\)
\begin{align}
|x_j - x_k| &\geq 2\delta, \quad \forall j, k = 1, 2, ...
\end{align}
\end{enumerate}
Now, denote by $K_{c,\delta}(x)$ a band-limited kernel with Fourier transform% , which has Fourier transform supported on an interval of length \(2\delta\) centered at a point \(c\)
% \[\supp\Bigl(\F[K_{c,\delta}]\Bigr)=[c-\delta, c+\delta].\] 
% The simplest example is the function \(K^*_{c,\delta}\) with Fourier transform
\begin{equation}\label{star1}
    \F[K_{c,\delta}](x) = \I \bigl(x \in [c-\delta, c+\delta]\bigr)
\end{equation}
for some central point $c\in \R$. Note that \(K_{c,\delta}\)  can be expressed explicitly as 
\begin{align}\label{star2}
    K_{c,\delta}(x) &= \frac{1}{2\pi}\int_{\R} \F[K_{c,\delta}](u) e^{-iux}du =  \frac{\sin(\delta x)}{\pi x} e^{-i c x}.
\end{align}
Let us take as the center point \(c\) values from an equidistant grid \(c_1, c_2,...\) with step \(2 \delta.\) Due to our assumption \hyperref[cond_A1]{(A1)}, for any central point \(c=c_j, j=1,2,...\) there exists at most one point from the set \(\mathcal{X}=\{-x_1, -x_2,...\}\), which belongs to the interval \( I_j = [c_j-\delta, c_j+\delta]. \) Denote \[\mJ:= \Bigl\{ j =1,2,...: \;\mX \cap I_j \ne \emptyset \Bigr\}.\]
Let us rearrange the points  from the set \(\mX\), and assign the index \(j \in \mJ\) to a point in the interval $I_j$. Correspondingly, \(p_j = \mu_d(x_j), j \in \mJ\). Under these notations, we have 
\begin{align}\label{eq37}
\F[K_{c_j,\delta}](-x_k) &= \I \bigl\{ j = k \bigr\} , \qquad j, k  = 1,  2, ...,
\end{align}
and, continuing the line of reasoning in \eqref{line1}, we arrive at
\begin{align}\label{eq1}
     (K_{c_j,\delta} * \phi_d)(u) &=  \begin{cases} p_j  e^{\i u x_j},& j \in \mJ,\\
     0,&j \notin \mJ.
     \end{cases}
\end{align}

% \subsubsection{Estimation of the discrete part} 
\item Joint consideration of \eqref{eq1} and \eqref{approx} leads to the conclusion that \(\p_j := \w p_j\)  can be approximated via the solution of the following optimization problem 
\begin{align}\label{opt}\breve{\p}_j &=\argmin_{\p_j} \int_{\R} w^{U_n} (u) 
\Bigl( \bigl| 
      K_{c_j,\delta} * \phi(u) \bigr| -  \p_j  \bigr)^2 du =      \int_{\R} w^{U_n} (u) \bigl| 
      K_{c_j,\delta} * \phi(u) \bigr| du,
      %=
%       \frac{\int_{\R} w^{U_n} (u) \bigl|
%      (K_{c_j,\delta} * \phi)(u) \bigr|  du}{\int_{\R} w^{U_n} (u)  du},
\end{align} 
where  $w^{U_n}(u) = U_n^{-1}w(u/U_n)$, \(w\) is a non-negative even weight function supported on \([-1, -\eps] \cup [\eps,1]\) with \(\eps \in (0,1)\), \(\int_\R w(u) du=1,\) and  \(U_n\) is an increasing sequence of positive numbers tending to infinity. If  the solution of~\eqref{opt} is not equal to zero, one can further represent \(x_j\) as the solution of 
\begin{align}\label{task2}
\breve{x}_j=\argmin_{x_j} &\int_{\R} w^{U_n} (u) \bigl( \Arg \bigl(      K_{c_j,\delta} * \phi(u) \bigr) +  2\pi k -u x_j  \bigr)^2 du,
\end{align} 
with any choice of the parameter \(k \in \Z\). Due to our assumptions on the weight function, we have \(\int_\R w^{U_n}(u) u du=0\), and therefore the solution of~\eqref{task2} doesn't depend on \(k\),\begin{align} \label{task3}
\breve{x}_j &=
\frac{\int_{\R} w^{U_n} (u) u \cdot
\Arg \bigl(      K_{c_j,\delta} * \phi (u) \bigr)
 du}{ \int_{\R} w^{U_n} (u) u^2 \, du} =\int_{\R} \widetilde{w}^{U_n} (u) \Arg\bigl( 
      K_{c_j,\delta} * \phi(u) \bigr) du,
 \end{align}
where \(\widetilde{w}^{U_n}(u) = w^{U_n} (u) u /\int_{\R} w^{U_n} (v) v^2 \, dv.\) Note that  this weight function has support on \(\{|u| \in [\eps U_n, U_n]\}\) and satisfies the property \( \widetilde{w}^{U_n}(u)=U_n^{-2}\widetilde{w}(u/U_n)  \) with  \(\widetilde{w}(\cdot)=\widetilde{w}^1(\cdot).\) 

%\begin{eqnarray}\label{opt}
%\inf_{\widetilde{p}_j, x_j} \int_{\R} w^{U_n} (u) \bigl| 
%\log\bigl(      (K_{c_j,\delta} * \phi)(u) \bigr) - \log\bigl( \p_j  \F[K_{c_j,\delta}](-x_j) \bigr) - \i u x_j \bigr|^2 du,
%\end{eqnarray} 
\end{enumerate}

\subsection{Estimation approach}  \label{secest1}
%In practice, it is convenient to choose some natural number \(J \geq 2\) such that  the measure \(\mu\) satisfies (A1) with  \(\delta= \bigl(X_{(n)} - X_{(1)}\bigr)/(2(J-1)).\) This choice leads to the division of the interval $I = [X_{(1)} - \delta, X_{(n)} + \delta]$ into \(J\) subintervals \(I_1,...,I_J\) of length \(2 \delta\).  
For simplicity we assume that the distribution $\mu_d$ has bounded support, i.e. $\supp(\mu_d) \subset [x_{\min}, x_{\max}]$ with some \(x_{\min}, x_{\max} \in \R\). Let us choose some natural number \(J\) such that  the measure \(\mu\) satisfies (A1) with  \(\delta= \bigl(x_{\max} - x_{\min}\bigr)/(2 J).\) Denote the central points $c_j = x_{\min} + (2j-1)\delta$ and corresponding intervals $I_j = [c_j - \delta, c_j + \delta]$ for $j = 1, \dots, J$.

To implement the ideas described in the previous section, we use a plug-in estimator based on the empirical characteristic function, 
\begin{align}\label{phin}
\widehat{\phi}(u) &= \frac{1}{n}
 \sum_{k=1}^n e^{\i u X_k},
\end{align}
leading to the estimates 
\begin{align}
\widehat{\p}_j &=
\int_{\R} w^{U_n} (u) \bigl| 
      K_{c_j,\delta} * \widehat{\phi}(u) \bigr| du, \qquad j=1,...,J. \label{est1}
\end{align}
Define \(\widehat{\mJ}=\widehat{\mJ}(p_\circ)=\{j=1,...,J: \; \widehat\p_j \geq \p_\circ\}\) with some small \(\p_\circ >0,\) which may depend on \(n\). For all indices \(j \in \widehat{\mJ},\) estimate \(\widehat{x}_j\) by 
      %\frac{\int_{\R} w^{U_n} (u) \bigl| 
%      (K_{c_j,\delta} * \widehat{\phi})(u) \bigr| du}{\int_{\R} w^{U_n} (u)   du}; \\
\begin{align}\widehat{x}_j &=
\int_{\R} \widetilde{w}^{U_n} (u) \Arg\bigl( 
      K_{c_j,\delta} * \widehat{\phi}(u) \bigr) du.
%\\
%\frac{\int_{\R} w^{U_n} (u) u \cdot
%\Arg\bigl(      (K_{c_j,\delta} * \widehat{\phi})(u) \bigr)
% du}{ \int_{\R} w^{U_n} (u) u^2 \, du}  \\ &=   U_n^{-1} \int_{\eps}^1 \widetilde{w}(u) \cdot
%\Im
%\log\bigl(      (K_{c_j,\delta} * \widehat{\phi})(u U_n) \bigr)
% du,
\label{est2}%, \qquad j=1,...,J
\end{align}
Practical implementation of \eqref{est1}-\eqref{est2} can be simplified by using the exact form of  \(K_{c_j,\delta} * \widehat\phi\), \begin{align}\label{convolution}
K_{c_j,\delta} * \widehat{\phi}(u) 
&=\frac{1}{n}
\int_{\R} K_{c_j,\delta}(v) \Bigl( 
\sum_{k=1}^n 
e^{\i (u-v) X_k}
\Bigr) dv =
\frac{1}{n}
\sum_{k=1}^n e^{\i u X_k} \F\bigl[K_{c_j,\delta}](-X_k) =\frac{1}{n} \sum_{-X_k \in I_j} e^{\i u X_k}.
\end{align}
%Next, we  decompose the  task~\eqref{opt} into the following two: 
Next, using that $\sum_{k=1}^{\infty} \p_k = \sum_{k=1}^{\infty} \w p_k = \w$,
we estimate \(\w\) by
\begin{align}\label{www}
\widehat{\w} =\sum_{j \in \widehat\mJ} \widehat{\p}_j.
\end{align}
Let us exclude the trivial cases, \(\widehat{\w}=0\) and \(\widehat{\w}=1.\) Define the estimates of \(p_j\) and \(\phi_d\) by
\begin{align}\label{phid}
 \widehat{p}_j=\frac{\widehat{\p}_j}{\widehat\w}, \;\; j \in \widehat\mJ, \qquad \text{and} \qquad \widehat{\phi}_d(u) = \sum_{j \in \widehat{\mJ}}\widehat{p}_j e^{iu \widehat{x}_j}.
\end{align} 
Finally, the estimates of the characteristic function of \(\mu_d\) and the density function \(g_{ac}\) are defined by \begin{align}\label{gac0}
\widehat{\phi}_{ac}(u) &=\frac{\widehat{\phi}(u)  - \widehat{\w} \widehat{\phi}_d(u)}{1-\widehat{\w}}, \qquad u \in \R,\\
\label{gac}
\widehat{g}_{ac}(x) &= \F^{-1} \Bigl[ \widehat{\phi}_{ac}(\cdot)\I\{|\cdot| \leq V_n\} \Bigr](x)
= \frac{1}{2\pi} \int_{-V_n}^{V_n}\frac{\widehat{\phi}(u)  - \widehat{\w} \widehat{\phi}_d(u)}{1-\widehat{\w}} e^{-\i u x}  du, 
\end{align}
where %$L$  is a positive regularising kernel supported on \([-1,1]\) and 
\(V_n\) is an unbounded increasing sequence of positive numbers.

\subsection{Convergence rates}
Define a class $\mathcal{S}(\w_{\min}, \w_{\max}, C, \delta)$ of distributions of the form  (\ref{first_model}), such that \(\w \in [\w_{\min}, \w_{\max}] \subset (0,1),\) the discrete part $\mu_d$ has bounded support and satisfies  (\hyperref[cond_A1]{A1}) with parameter \(\delta>0\), and the absolutely continuous part \(\mu_{ac}\) has a differentiable density function \(g_{ac}\), which satisfies the following conditions:
\begin{equation*}
  \max_{x \in \R}\bigl\{
| g_{ac}(x)|, | g_{ac}'(x)|\bigr\} < C. %\qquad\int_{|x|>1 }|x| d\mu(x) < C.
\end{equation*}
%Before we formulate the result showing the convergence rates for this class, let us recall that $w^{U_n}(u) = U_n^{-1}w(u/U_n)$, where \(w\) is a non-negative even weight function supported on the interval \([-1, -\eps] \cup [\eps,1]\) with \(\eps \in (0,1)\), \(\int_\R w(u) du=1\), and  \(U_n\) is an increasing sequence of positive numbers tending to infinity.
\begin{theorem}\label{thm1}
%Let $w(u)$ be a non-negative weight function with \(w \in L^1([\eps,1] )\) with  \(\eps>0\) and support on $[\eps,1]$. Let distribution $\mu \in \mathcal{S}(p_{min}, C)$ for some positive constants $p_{min}$, $C$. Hence, under these assumptions and notations the following convergence rates hold for estimators $\widehat{p}_j$ and $\widehat{x}_j$ for the model \ref{main_model} for $n \rightarrow \infty$
Let distribution $\mu \in \mathcal{S}(\w_{\min}, \w_{\max}, C,\delta)$ for some positive constants $\w_{\min}, \w_{\max} \in (0,1)$, $C>0$, \(\delta>0\). Denote
\begin{equation*}
    \mQ_n =
\frac{\sqrt{\log(nU^2_n)}}{\sqrt{n}} + \frac{\sqrt{\delta}C}{\sqrt{U_n}}.
\end{equation*}
Then 
\begin{align}\label{ratepj}
    \sup_{\mathcal{S}} \max_{j} \bigl| \widehat{\p}_j - \p_j \bigr| &= O_{\mathbb{P}}\Bigl(\mQ_n\Bigr).
    \end{align}
If \(\mQ_n=o(1)\) and \(\p_\circ= c U_n^{-1}\) for any \(c>0\), then  
    \begin{align*} 
\sup_{\mathcal{S}} \max_{j \in \mJ \cap \widehat{\mJ}}|\widehat{x}_j - x_j| =   O_{\P} \Bigl( 
\mQ_n        \Bigr), \qquad
\sup_{\mathcal{S}} |\widehat{\w} - \w| =   O_{\P} \Bigl( 
\mQ_n        \Bigr), \qquad 
    \sup_{\mathcal{S}} \max_{j} \bigl| \widehat{p}_j - p_j \bigr| = O_{\mathbb{P}}\Bigl(\mQ_n\Bigr).
\end{align*}
Under the choice $U_n = n$ we have \(
    \mQ_n=O\Bigl(\sqrt{\log n} / \sqrt{n}\Bigr)
\), and therefore the convergence rates are parametric up to a logarithmic factor.
\end{theorem}

%The next result gives convergence rates for the parameters of the mixture, which estimators depend on the estimation of $\w$

%\begin{remark}
%    The results in Theorem $\ref{thm1.1}$ are similar to the rates obtained by Theorem $\ref{thm1}$ since $J_n$ naturally depends on $\delta$, because intervals of length $l < 2\delta$ provide no additional inference data comparing to intervals of length $l = 2\delta$ and therefore $J_n$ is bounded in some way.
%\end{remark}

Next, we proceed to the convergence rates of $\widehat{g}_{ac}$. Not surprisingly, the rates depend on further assumptions on the class of densities \(g_{ac}\). Below we consider two particular cases, which are very common in statistical literature (see, e.g., \cite{Meister}), namely, the classes of ordinary smooth and supersmooth densities,
\begin{align*}
\mathcal{P}=\mathcal{P}_{\alpha, \mathtt{C}}&:= \Bigl\{g_{ac}:  |\phi_{ac}(u)| \leq \mathtt{C} \bigl(1+|u|\bigr)^{-\alpha}\Bigr\}, \qquad \alpha, \mathtt{C}>0,\\
\mathcal{E}=\mathcal{E}_{\gamma,\mathtt{C}}&:= \Bigl\{g_{ac}: \; |\phi_{ac}(u)| \leq \exp\{-\mathtt{C} |u|^\gamma\} \Bigr\}, \qquad \gamma,\mathtt{C}>0.
\end{align*}

\begin{theorem}\label{thm2}
Let the assumptions of Theorem~\ref{thm1} be fulfilled. Then 
\begin{align}\label{rateG}
\sup_{\mathcal{S}} \sup_{x \in \R} | 
 \widehat{g}_{ac}(x) -  g_{ac}(x) 
 | 
 = O_{\P}\Bigl( 
 V_n    \sqrt{
\log(n V_n^2) / n
}
 +V_n^2\mQ_n+V_n\,\mQ_n
+
 \int_{|u|>V_n} |\phi_{ac}(u)|du
 \Bigr).
\end{align}
In particular, if \(g_{ac} \in \mathcal{P}_{\alpha,\mathtt{C}}\) for some \(\alpha > 3\), \(\mathtt{C}>0\), then the choice \(U_n =n,\) \(V_n = n^{1/(2(\alpha-1))}\) leads to 
\begin{align}\label{rateP}
\sup_{\mathcal{S} \cap \mathcal{P}} \sup_{x \in \R} | 
 \widehat{g}_{ac}(x) -  g_{ac}(x) 
 | 
 = O_{\P}\Bigl( 
 \frac{\sqrt{\log(n)}}{n^{(\alpha-3)/(2(\alpha-1))}}
 \Bigr),
\end{align}
that is, the rate is polynomial but not parametric. If \(g_{ac} \in \mathcal{E}_{\gamma,\mathtt{C}}\) for some \(\gamma > 2\), \(\mathtt{C}>1,\) then the choice \(U_n =n,\) \(V_n = \sqrt{\log(n)}\) leads to 
\begin{align}\label{rateE}
\sup_{\mathcal{S}\cap{\mathcal{E}}} \sup_{x \in \R} | 
 \widehat{g}_{ac}(x) -  g_{ac}(x) 
 | 
 = O_{\P}\Bigl( 
 \frac{(\log n)^{3/2}}{\sqrt{n}}
 \Bigr),
\end{align}
and therefore the convergence rates are parametric up to a logarithmic factor.
\end{theorem}

\section{Inference for quasi-L{\'e}vy measures} \label{secqid}

\subsection{Estimation approach}\label{secest2}
This section discusses the case when the model of type~\eqref{first_model} belongs to the class \(\Q\). Recall that Proposition~\ref{propmain}(i) states that  the characteristic triplet is equal to \((\gamma,0, \nu)\) with some \(\gamma>0\) and  the quasi-L{\'e}vy measure \(\nu\) in the form~\eqref{nu1}. Note also that due to Theorem 2.2 from \cite{BK2023} the condition $\mu \in \Q$ yields $\mu_d \in \Q$ with quasi-L{\'e}vy measure \(\nu_d(dx)=\sum_{y \in C} c_y \delta_y (dx)\). In what follows we aim to recover the quasi-L{\'e}vy measure \(\nu\) and its two components: \(\nu_d\) and the second component, which may not be a quasi-L{\'e}vy measure of any other distribution. 

For the estimation of \(\nu\), we represent the characteristic exponent \(\psi(u) := \log\bigl( \phi(u) \bigr)\) using the L{\'e}vy-Khintchine representation~\eqref{LC} with \(\sigma=0,\)
\begin{align*}
\psi(u) &=
 \i \gamma u
		    +
		    \int_{\mathbb{R}\setminus \{0\}}
		    \left( 
		    		e^{\i u x}-1 - \i u x \I\{x \in [-1,1]\}		\right)			\nu(dx),
\end{align*} and take the second derivative of both sides of the last equality, 
\begin{equation}\label{imp_eq}
    \psi''(u) = \int (ix)^2e^{iux} \nu(dx) = -\F[\bar{\nu}](u),
\end{equation}
where \(\bar{\nu}(dx) := x^2 \nu(dx).\) Inspired by (\ref{imp_eq}), we introduce an estimator 
\begin{equation}\label{barnuu}
   \widehat{\bar{\nu}} = -\F^{-1}\Bigl[\widehat\psi'' \F[K_n]\Bigr] = -\F^{-1}\Bigl[(\widehat{\phi}''/\widehat{\phi} - (\widehat{\phi}'/\widehat{\phi})^2)\F[K_n]\Bigr],
\end{equation}
where $\widehat{\phi}(\cdot)$ is an empirical characteristic function~\eqref{phin}, $\widehat{\psi}(\cdot) := \log \widehat{\phi}(\cdot)$, $K_{n}(\cdot) := W_nK(\cdot W_n)$ with an unbounded increasing sequence $W_n$ and a positive band-limited kernel $K$,  such that
\begin{equation*}
    \int K(x)dx = 1, \quad  \int |x|^{1/2}K(x)dx < \infty, \quad \supp\Bigl(\F[K]\Bigr) \subseteq [-1, 1].
\end{equation*}
The term $\F[K_n]$ is added to the estimator to restrict $\widehat{\psi}''(u)$ to the interval $[-W_n, W_n]$ and to smooth the inverse Fourier transform. The same approach may be applied to the discrete part of the mixture where the estimate of $\psi_d$ is based on~\eqref{phid},
\begin{equation*}
    \widehat{\psi}_d := \log\widehat{\phi}_d(u)= \log \Bigl( \sum_{j\in \widehat\mJ}\widehat{p}_j e^{iu\widehat{x}_j} \Bigr).
\end{equation*}

\subsection{Convergence rates}
Following ideas from~\cite{BR2015}, we will 
%We will show the convergence rates of the propos
%Throughout this section we use the standard norm in the $L^2$ space, 
%\begin{equation*}
%    ||f|| = \Bigl( \int_{\R} |f(x)|^2 dx \Bigr)^{1/2}, \quad f\in L^2.    
%\end{equation*}
%Further, we will 
derive the convergence rates in terms of the operator norm in the space $H^{-1}(\R)$, which is a dual space to the Sobolev space $H^1(\R)$. Recall that the one-dimensional space $H^1(\R)$ consists of the functions $f \in L^2(\R)$ with $f' \in L^2(\R)$. The norm in  $H^1(\R)$  is defined as
\begin{equation*}
    ||f||_{H^1} = %(\int_{\R} |f(x)|^2 dx + \int_{\R} |f'(x)|^2 dx)^{1/2} = 
    (||f||^2 + ||f'||^2)^{1/2},
\end{equation*}
where \(||f||=\bigl(\int_{\R} |f(x)|^2 dx\bigr)^{1/2}\). 
For an operator $G \in H^{-1}(\R)$, define the norm \begin{equation}\label{first_norm}
    ||G||_{H^{-1}} = \frac{1}{\sqrt{2\pi}}||(1 + u^2)^{-1/2}\F[G](u)||, \quad G \in H^{-1},
\end{equation}
which coincides with the standard definition of the norm in dual spaces which is
\begin{equation}\label{second_norm}
    ||G||_{H^{-1}} = \sup_{||f||_{H^1} = 1} |G \circ f|, \quad G \in H^{-1}.
\end{equation}
For positive constants \(C_1, C_2>0,\) introduce the class 
\begin{equation}\label{class_M}
    \mathcal{M}(C_1, C_2) := \Bigl\{ \mu \in \Q : \;\; \E_{\mu}\, |X|^4 \leq C_1, \quad |\bar{\nu}|(\R) \leq C_2 \Bigr\},
\end{equation}
where \(X \sim \mu,\) \(\nu\) is the quasi-L\'evy measure of \(\mu,\) and \(\bar{\nu}(dx) = x^2 \nu(dx)\).

\begin{theorem}\label{thm3}
     The estimator $\widehat{\bar{\nu}}$ has the following convergence rate
    \begin{equation*}
        \sup_{\mathcal{S}\, \cap\, \mathcal{M}(C_1, C_2)} ||\widehat{\bar{\nu}} - \bar{\nu}||_{H^{-1}} = O_{\mathbb{P}}\Bigl(C_1^{1/2}\Bigl(\frac{\log (nW^2_n)}{n}\Bigr)^{1/4} + C_2\, W_n^{-1/2}\Bigr), \quad n \rightarrow \infty,
    \end{equation*}
    provided \(W_n \lesssim n.\) The choice $W_n = n$ yields polynomial convergence rate
    \begin{equation*}
        \sup_{\mathcal{S}\, \cap\, \mathcal{M}(C_1, C_2)}||\widehat{\bar{\nu}} - \bar{\nu}||_{H^{-1}} = O_{\mathbb{P}}\Bigl(C_1^{1/2}\,n^{-1/4}(\log n)^{1/4}\Bigr), \quad n \rightarrow \infty.
    \end{equation*}
\end{theorem}
\begin{remark}
The assumption on the fourth moment of $\mu$ 
is restrictive, but it is also referenced in other studies that analyse convergence rates of estimates, which are based on the derivatives of empirical characteristic functions, see, e.g, \citep{NeumannReiss2009}. 
\end{remark}
\begin{remark}
A similar result for infinitely divisible distributions can be found in \citep{BR2015}, Proposition~6.5. After this proposition, the discussion shows that the rate of convergence can even be  parametric \(O_{\P}(n^{-1/2}),\) but it essentially depends on the asymptotic properties of the quantity 
\begin{align*}
M_h := \max_{k=0,1,2} \sup_{|u| \leq 1/h} \bigl| 
(1/\phi)^{(k)}(u)\bigr|.
\end{align*}
However, as the authors write, the behavior of \(M_h\) ``is unknown to the statistician''. In this regard, Theorem~\ref{thm3} has an advantage because it provides a uniform upper bound for the entire class of rational-infinitely divisible distributions from the class \(\mathcal{M}(C_1, C_2)\).  
\end{remark} 

Next we provide a similar result on the convergence rates for $\widehat{\bar\nu_d}.$

\begin{theorem}\label{thm4}
Let the assumptions of Theorems~\ref{thm1} and  \ref{thm3} be fulfilled. Then 
\[
\sup_{\mathcal{S}\,\cap\,\mathcal{M}(C_1,C_2)}\|\widehat{\bar\nu}_d-\bar\nu_d\|_{H^{-1}}
= O_\P\Bigl( W_n \mQ_n + C_2\, W_n^{-1/2}\Bigr).
\]
The choice $U_n = n$, $W_n = C_2^{2/3}\bigl(\tfrac{n}{\log n}\bigr)^{1/3}$ balances the terms and  yields the polynomial rate
\[
\sup_{\mathcal{S}\,\cap\,\mathcal{M}(C_1,C_2)}\|\widehat{\bar\nu}_d-\bar\nu_d\|_{H^{-1}}
= O_\P\Bigl( C_2^{2/3}\, n^{-1/6}(\log n)^{1/6}\Bigr).
\]
% The choice $U_n = n/\log n$ balances the terms and yields the polynomial rate
% \[
% \|\widehat{\bar\nu}_d-\bar\nu_d\|_{H^{-1}}
% =O_\P\Bigl((\log n/n)^{1/4}\Bigr).
% \]
% which is faster than the rate $O_\P(n^{-1/6}\sqrt{\log n})$ of Theorem~5 (the
% full quasi-L\'evy measure). The simpler choice $U_n=n$ gives
% $O_\P(n^{-1/4}\sqrt{\log n})$.
\end{theorem}

\begin{remark}
%The estimator of the continuous part $\bar\nu_{ac}$ is then obtained for free as
%$\widehat{\bar\nu}-\widehat{\bar\nu}_d$, recovering the component that does not
%correspond to the discrete law. Note also that, unlike Theorem \ref{thm3}, no finite
%fourth moment of $\mu$ is required here: $\widehat\phi_d$ is differentiated
%analytically, so only $|\bar\nu_d|(\R)<\infty$ enter.
The estimator for the continuous part, $\bar{\nu}_{ac}$ is then obtained as the difference $\widehat{\bar{\nu}} - \widehat{\bar{\nu}}_d$, recovering the component that does not correspond to the discrete law. %Note that, unlike in Theorem \ref{thm3}, a finite fourth moment of $\mu$ is not used for deriving the rates of convergence.
\end{remark}

\section{Numerical results}\label{numerical_results}
In the examples below we will consider the following mixture
\begin{equation}\label{example_model}
    \mu = \w \mu_{Pois} + (1 - \w)\mu_{Exp},
\end{equation}
where $0 < \w < 1$, $\mu_{Pois}\sim\text{Pois}(\lambda_1),~\mu_{Exp}\sim\text{Exp}(\lambda_2)$ with $\lambda_1,\lambda_2>0$. 
Given the observations $X_1, \dots, X_n$ of $\mu$, we aim to illustrate the algorithms presented in Sections~\ref{secest1} and \ref{secest2}. We proceed in three steps.
\begin{enumerate}
    \item Analyze the estimation quality of the discrete part, that is, \(\mu_{Pois}(k)\) for \(k=0,1,2,..\) and the parameter $\lambda_1$, and the mixture parameter $\w$.
    \item Analyze the nonparametric estimate for  $\mu_{Exp}$.
    \item Prove that under certain conditions on $\lambda_1$, $\lambda_2$, $\w$, distribution $\mu \in \Q$, and analyze the estimate for  $\F[x^2\nu]$, where \(\nu\) is a quasi-L{\'e}vy measure. The choice of this estimation target is motivated by Proposition~\ref{propmain}(ii), which leads to the closed-form expression for this object.
\end{enumerate}

\subsection{Estimation results for the discrete part}

The Poisson distribution satisfies assumption \hyperref[cond_A1]{(A1)} with $\delta = 1/2$. % Next, we fix the interval $[x_
%{\min}, x_{\max}]$ which we will then split into subintervals of length $2\delta$. 
%Although the discrete part in models of type (\ref{first_model}) may have unbounded support (as in our case), we choose an interval $[x_
%{\min}, x_{\max}]$ that covers most of the probability mass of the discrete distribution. It is convenient to fix  $x_{\min}=X_{(1)} - \delta, x_{\max}= X_{(n)} + \delta$, and 
Define the subintervals $I_j = [(2j - 3)\delta, (2j - 1)\delta]$, $j = 1, \dots, J$ for \(J=20\). Note that in our case the discrete part of the model (\ref{first_model})  has unbounded support. We choose an estimation interval  that covers most of the probability mass of the discrete distribution.

For this numerical study we fix $\w  = 1/5$, $\lambda_1 = 3$, $\lambda_2 = 3$.  Figure~\ref{fig1} shows the histogram of the sample drawn from \(\mu\) with barwidth equal to 1.  Note that the discrete part cannot be directly revealed from the histogram, as it is for mixture distributions with large \(\w\) and  the discrete part, supported on a small number of points (for example, the  contamination model~\eqref{second_model} with a Bernoulli distribution).

 Figure~\ref{fig2} provides boxplots for estimates~\eqref{est1} of  $p_j=\mu_{Pois}(j), j=1,2,3,4,$ for $N=20$ simulation runs and $n=1000, 2500, 5000$. The boxplots demonstrate the reduction of estimation error as the number of observations increases.  Similar quality of estimation holds for the estimate~\eqref{www} of $\w$ and the estimate~\eqref{est1} for $\lambda_1=-\ln p_0$, as it is shown in Figure~\ref{fig3}.

\begin{figure}[ht]
    \centering
    \includegraphics[scale=0.4]{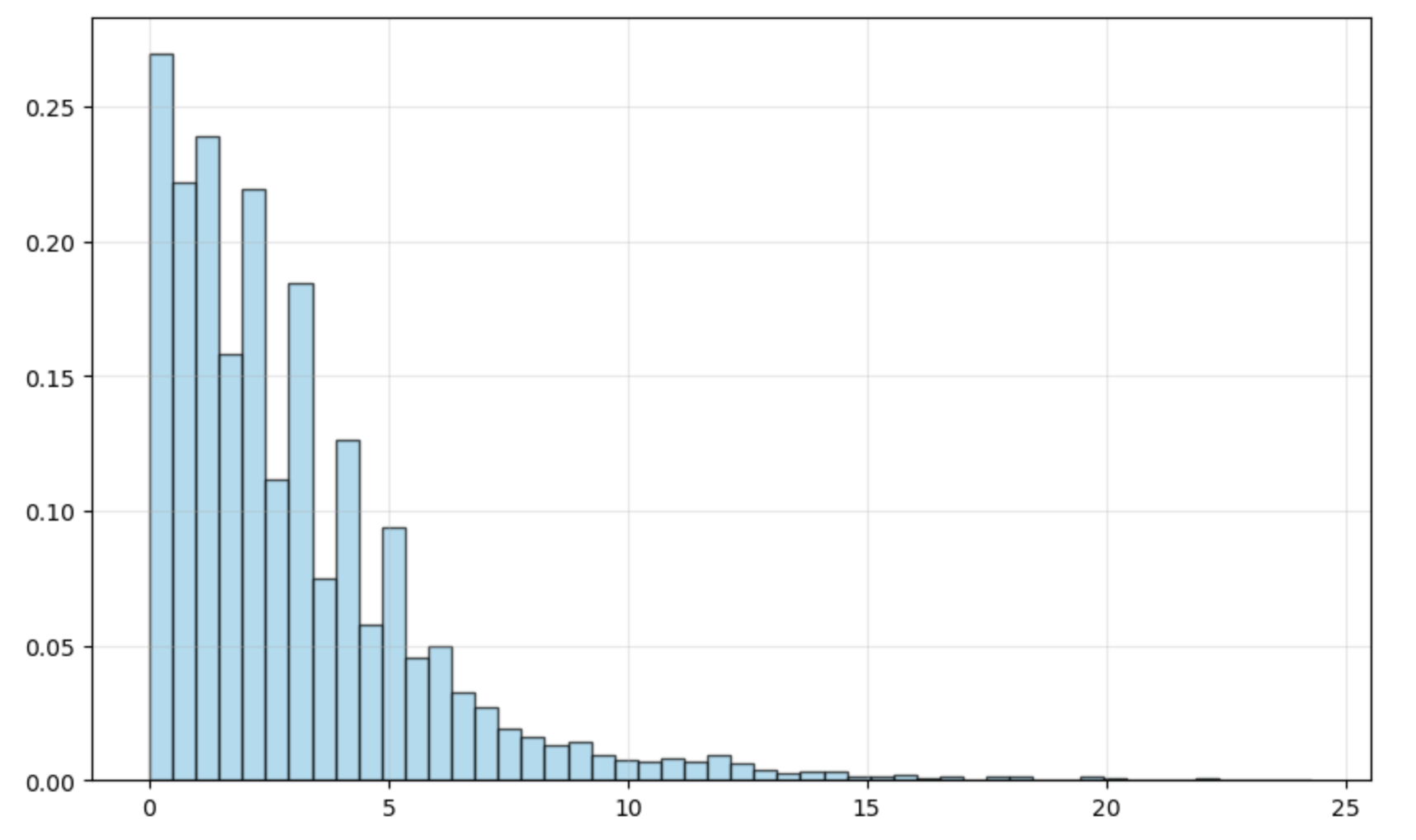}
    \caption{\label{fig1}  Histogram of the mixture distribution $\mu$.}
\end{figure}

% For the current example we fix $\pi = p = 0.3$ and atoms of $\mu_d$ at points $x_0 = 0$, $x_1 = 1$. Figure~\ref{fig1} shows the histogram of the described distribution.

%\begin{figure}[ht]
%    \centering
%    \includegraphics[scale=0.39]{images/p_j_boxplots_new.png}
%    \caption{\label{fig2}  Boxplots for the estimates of  $p_j, j=1,2,3,4$.}
%\end{figure}

%\begin{figure}[ht]
%    \centering
%        \includegraphics[scale=2]{images/ppp.png}
%%    \includegraphics[width=0.44\textwidth]{images/p_boxplots.png}
% %   \hfill
%  %  \includegraphics[width=0.54\textwidth]{images/lam1_boxplots.png}
%    \caption{\label{fig3}Boxplots for the estimates of $\w$ and $\lambda_1$.}
%\end{figure}

\subsection{Estimation results for the absolutely continuous part}

In this subsection we analyze the estimation of the absolutely continuous part of the mixture~\eqref{example_model}.   % as follows
% \begin{equation*}
%    \widehat{\phi}_{Exp}(u) :=  \frac{\widehat{\phi}(u) - \widehat{\w} \widehat{\phi}_{Pois}(u)}{1 - \widehat{\w}}.
% \end{equation*}
First plot in Figure~\ref{fig4} compares the real part of the estimate $\widehat{\phi}_{Exp}(u)$, defined by~\eqref{gac0}, with the  real part of the true characteristic function equal to $\Re \phi_{Exp}(u)  = \lambda_2^2 / (\lambda_2^2 + u^2)$. Second plot in Figure~\ref{fig4} compares the estimate~\eqref{gac} with the true density of the exponential term $g_{Exp}(x) = \lambda_2 e^{-\lambda_2 x}, x>0$. Both plots demonstrate the good performance of the proposed estimators. 
\begin{figure}[ht]
    \centering
    \includegraphics[scale=0.4]{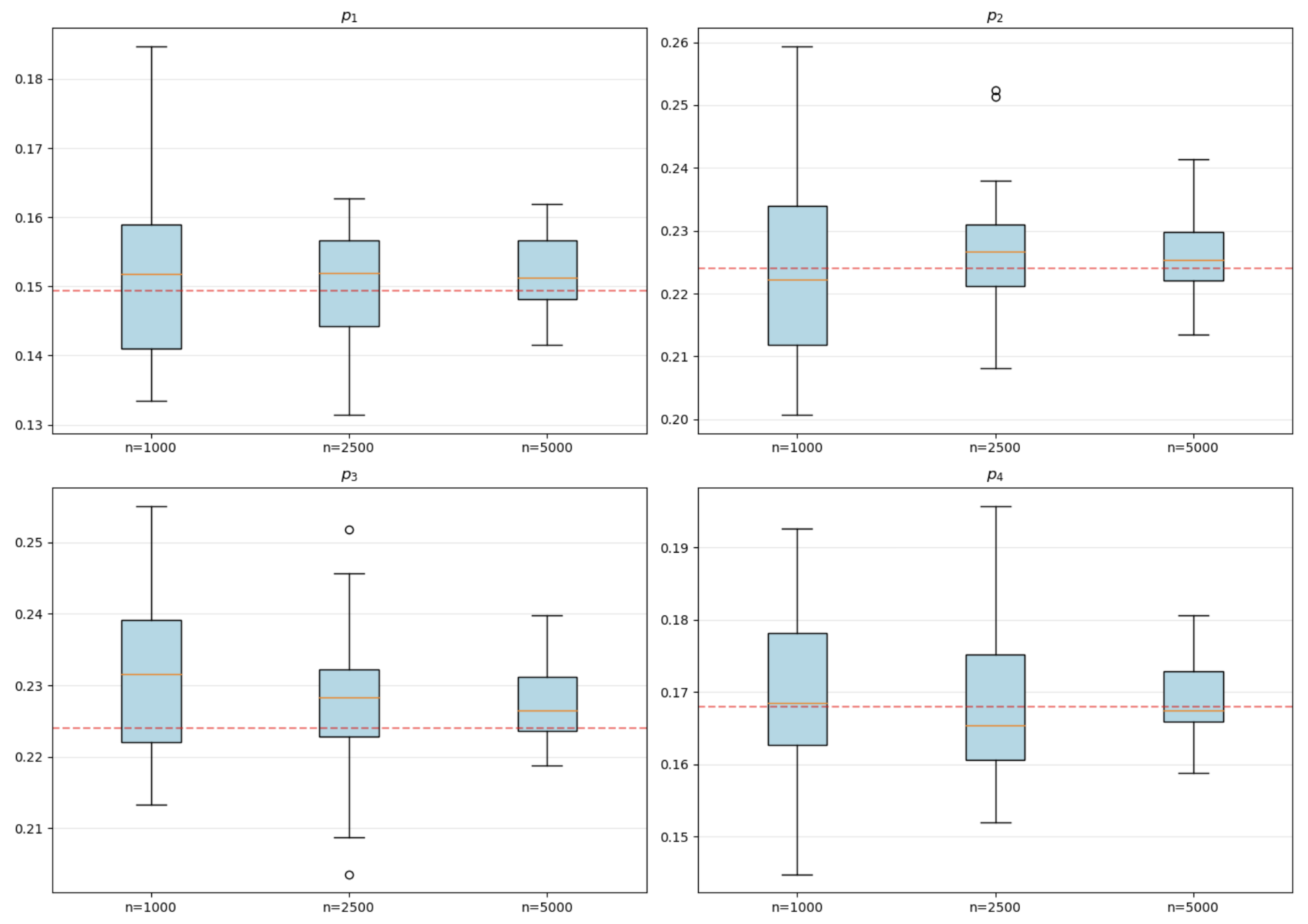}
    \caption{\label{fig2}  Boxplots for the estimates of  $p_j, j=1,2,3,4$.}
\end{figure}

\begin{figure}[ht]
    \centering
    \includegraphics[scale=0.4]{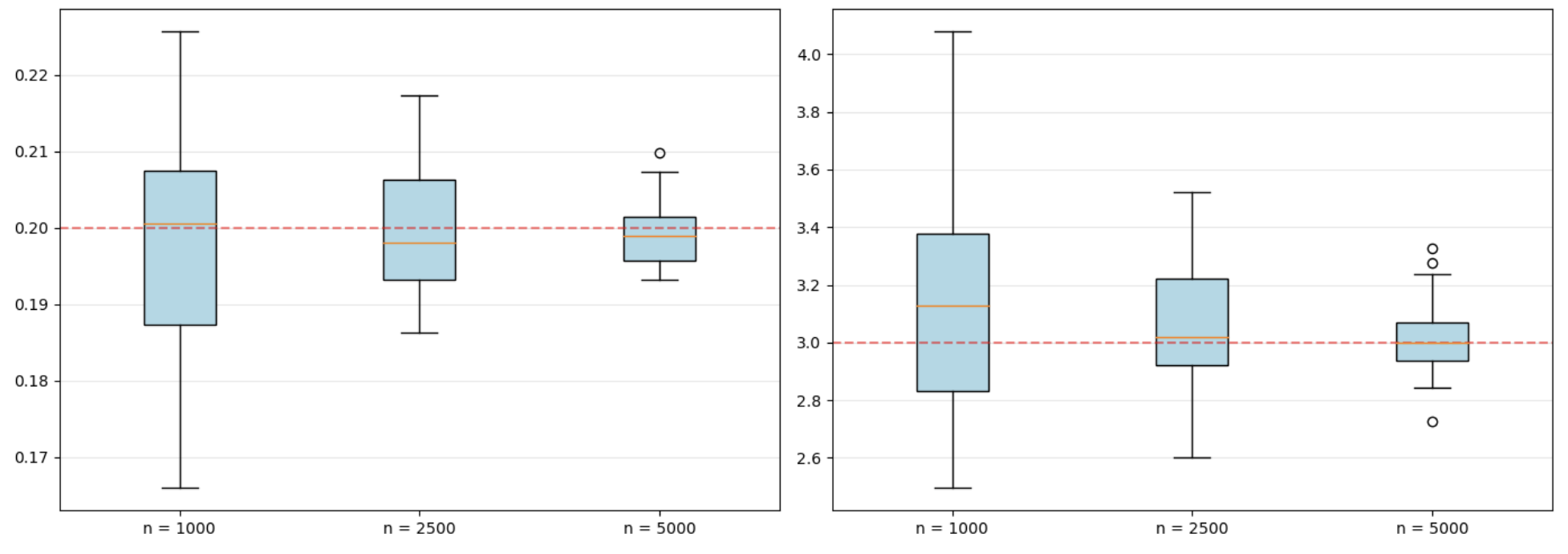}
    \caption{\label{fig3}  Boxplots for the estimates of $\w$ and $\lambda_1$.}
\end{figure}

\begin{figure}[h]
    
    \centering
    
    % Ряд 1: два изображения
    \includegraphics[width=0.49\textwidth]{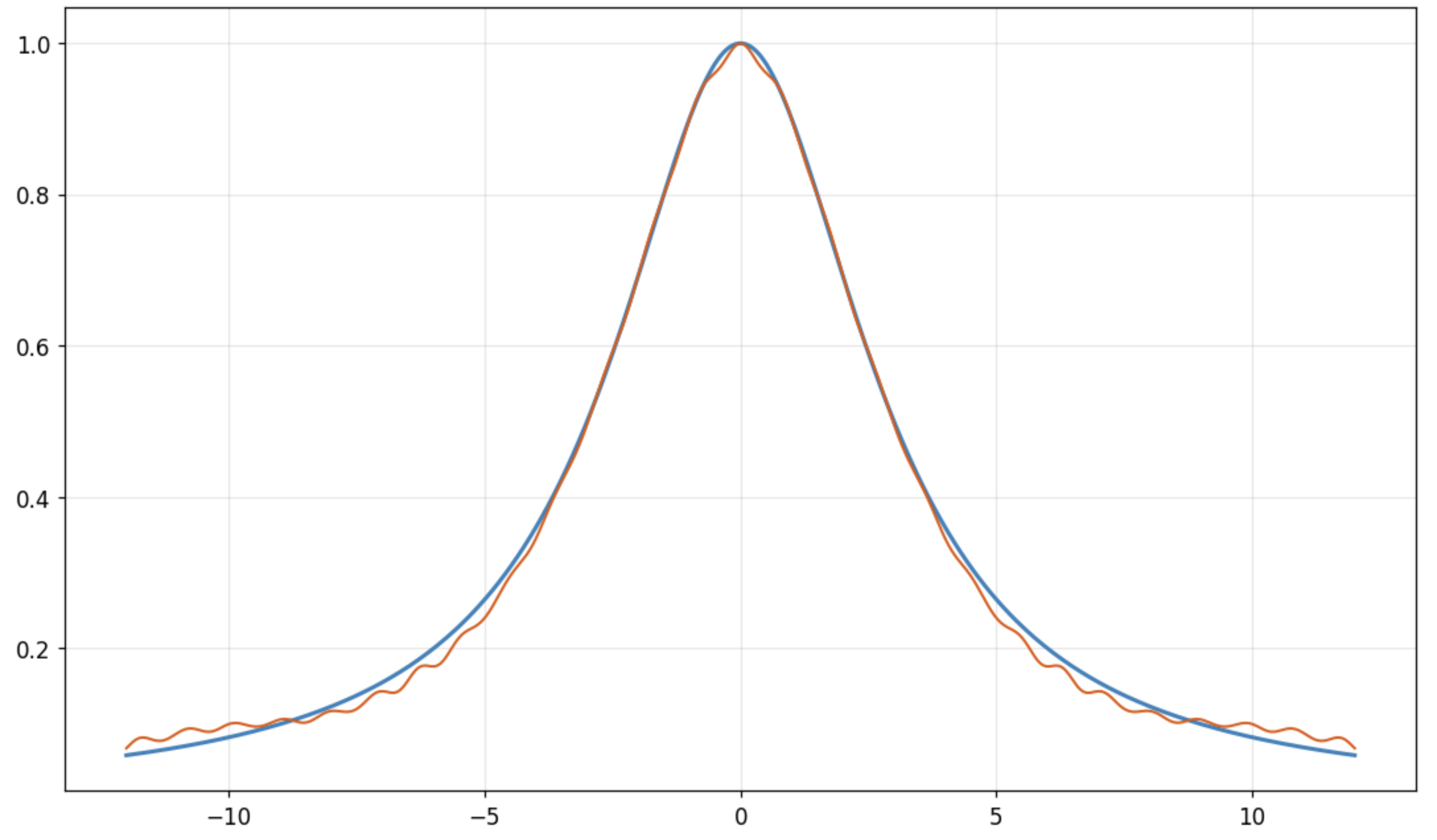}
    \hfill
    \includegraphics[width=0.49\textwidth]{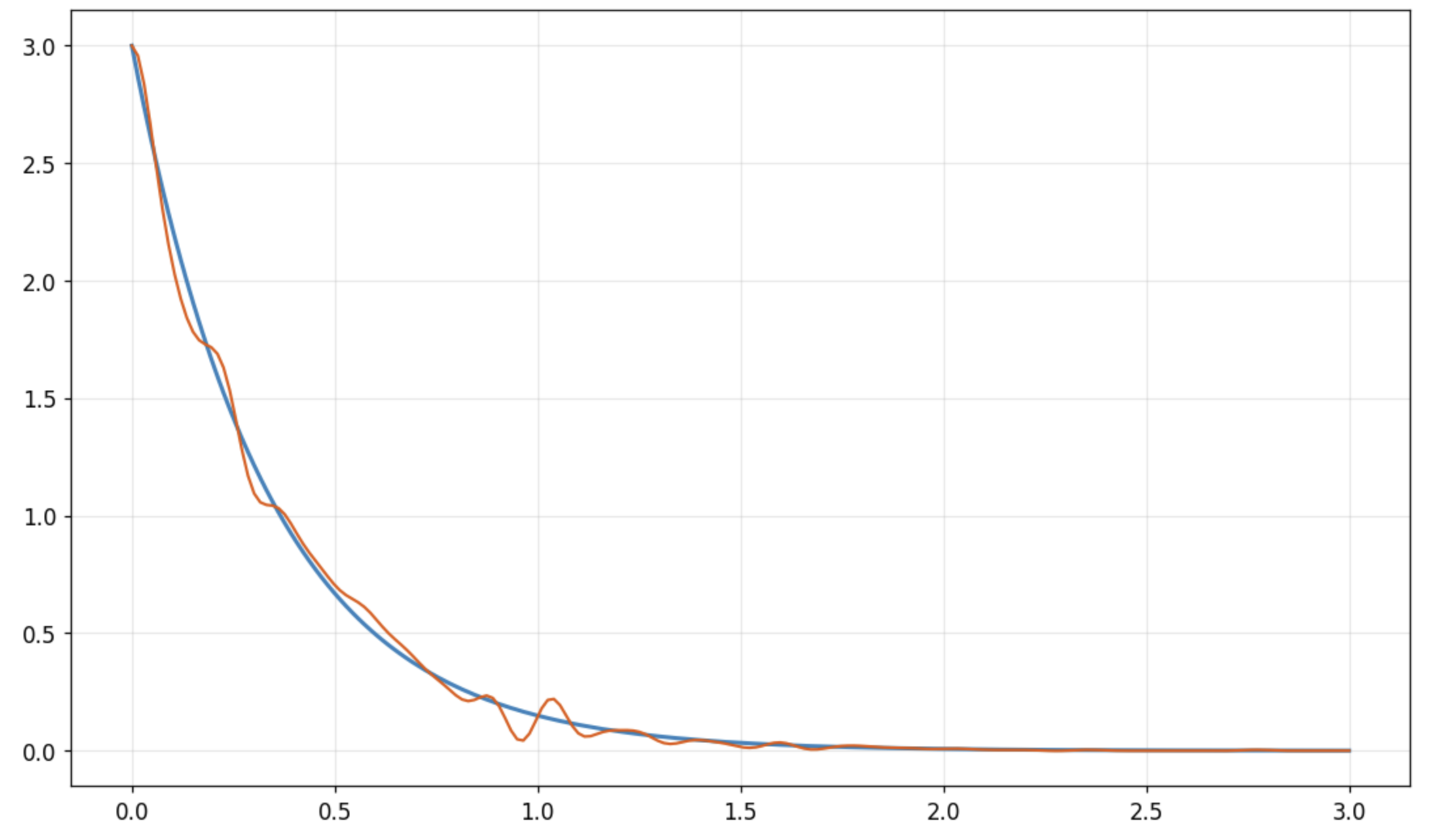}
    
     % \vspace{1em}
    
    % % Ряд 2: одно изображение по центру
    % \includegraphics[width=0.45\textwidth]{images/g_exp_5000_smooth.png}
    
    \caption{\label{fig4} Plots of the true functions $\Re\phi_{Exp}$, $g_{Exp}$ (blue lines) and their estimates $\Re \widehat{\phi}_{Exp}$, $\widehat{g}_{Exp}$ (orange lines) based on \(n=5000\) observations.}
\end{figure}

\subsection{Estimation results for the quasi-L{\'e}vy measure}

In this subsection we consider the case when $\mu \in \Q$ and use Proposition~\ref{propmain}(ii), which gives the exact form of its quasi-L{\'e}vy measure.  The next lemma plays an essential role. \begin{lemma}\label{example_lemma}
    Let $\mu$ be a distribution of the form (\ref{example_model}) with $\w > 1/2$ and assume that 
\begin{align}\label{ccc}
\lambda_1 < \frac{1}{2} \ln \frac{\w}{1 - \w}.
\end{align} 
Then $\mu \in \Q$ and its quasi-L{\'e}vy measure is equal to 
    \begin{equation*}
        \nu = \nu_{Pois} + \sum_{m = 1}^{\infty} \frac{(-1)^{m + 1}}{m}\Bigl( 
\frac{1-\w}{\w} 
\Bigr)^m\Lambda^{*m},
    \end{equation*}
where the signed measure $\Lambda$ is defined by
\begin{equation}\label{Lambda}
    \Lambda(dx) = \lambda_2e^{\lambda_1} \sum_{k=0}^{\infty} \frac{(-\lambda_1)^k}{k!}\exp(-\lambda_2(x - k)) \I_{[k, \infty)}(x)dx.
\end{equation}
\end{lemma}
\begin{proof}
Following  Proposition~\ref{propmain}, we consider the ratio between the Fourier transforms of  \(\phi_{Exp}\) and \(\phi_{Pois}\), 
\begin{align*}
        %\F\bigl[\Lambda \bigr](u) &=
        \frac{\phi_{Exp}(u)}{\phi_{Pois}(u)} = \frac{\lambda_2}{\lambda_2 - i u}\exp\{\lambda_1 (1 - e^{iu})\} =e^{\lambda_1} \sum_{k=0}^{\infty} \frac{(-\lambda_1)^k}{k!}\frac{ \lambda_2  e^{iuk}}{\lambda_2 - iu}.
    \end{align*}
    Note that \(\lambda_2  e^{iuk}/ (\lambda_2 - iu)= \phi_{k}(u) \phi_{Exp}(u),\)
    where $\phi_k(u)$ is a characteristic function of a constant $k$, and 
    \begin{equation*}
    \F^{-1}[ \frac{\lambda_2e^{iuk}}{\lambda_2 - iu}] = \lambda_2 \e^{-\lambda_2 (x - k)} \I_{[k, \infty)}(x).
    \end{equation*}
    This observation leads to the conclusion that  \(\phi_{Exp}(u)/\phi_{Pois}(u) = \F[\Lambda](u)\), where \(\Lambda\) is given by~\eqref{Lambda}. We have 
    \begin{align*}
    |\Lambda|(\R) &= \int_{\R} |e^{\lambda_1} \sum_{k=0}^{\infty} \frac{(-\lambda_1)^k}{k!}\lambda_2\e^{-\lambda_2(x - k)} \I_{[k, \infty)}(x)|dx 
%    &\leq \int_{\R} e^{\lambda_1} \lambda_2 \e^{-\lambda_2 x}\sum_{k=0}^{\infty} \frac{(\lambda_1)^k}{k!}\e^{\lambda_2 k} \mathbb{I}_{[k, \infty]}(x) dx \\
    \leq e^{\lambda_1} \lambda_2 \sum_{k=0}^{\infty} \frac{\lambda_1^k}{k!} \Bigl( \int_{k}^{\infty}e^{-\lambda_2 (x-k)} dx\Bigr) =  e^{\lambda_1}  \sum_{k=0}^{\infty} \frac{\lambda_1^k}{k!} = e^{2\lambda_1},
    \end{align*}
    leading to the condition~\eqref{ccc}.
\end{proof}
In what follows, we fix the parameters $\lambda_1 = 0.4, \lambda_2 = 0.5, \w = 0.75$, which satisfy the condition~\eqref{ccc}.  The Fourier transform of $\bar{\nu}$ is equal to 
\begin{equation*}
    \F[\bar{\nu}](u) =  \F[x^2\nu](u) = \F[x^2\nu_{Pois}](u) + \F[x^2\tilde{\nu}](u),
\end{equation*}
where \(
\F[x^2\nu_{Pois}] (u) = -\frac{d^2}{du^2} \ln(\phi_{Pois}(u)) = \lambda_1 \e^{\i u},\)
and
\begin{align*}
    \F[x^2\tilde{\nu}] (u)&= \F[x^2\sum_{m = 1}^{\infty} \frac{(-1)^{m + 1}}{m}\Bigl( 
\frac{1-\w}{\w} 
\Bigr)^m\Lambda^{*m}] (u)= -\frac{d^2}{du^2} \sum_{m = 1}^{\infty} \frac{(-1)^{m + 1}}{m}\Bigl( 
\frac{1-\w}{\w} 
\Bigr)^m\F[\Lambda^{*m}](u) \\
&= -\frac{d^2}{du^2} \sum_{m = 1}^{\infty} \frac{(-1)^{m + 1}}{m}\Bigl( 
\frac{1-\w}{\w} 
\Bigr)^m\Bigl(\frac{\phi_{Exp}(u)}{\phi_{Pois}(u)}\Bigr)^m = -\frac{d^2}{du^2} \ln\Bigl(1 + \frac{1 - \w}{\w}\frac{\phi_{Exp}(u)}{\phi_{Pois}(u)}\Bigr).
\end{align*}
This expression allows us to compare $\F[\bar{\nu}]$ with its estimator, equal to $-\widehat\psi''(u) \F[K](u)$, see~\eqref{barnuu}. For this example we choose $K(x)$ as
\begin{equation*}
    K(x) = \frac{1}{2\pi}\Bigl( \frac{\sin(x/2)}{x/2} \Bigr)^2,
\end{equation*}
with Fourier transform equal to  \(
    \F[K](u) =
        \bigl( 1 - |u| \bigr) \I\bigl\{  |u| \leq 1\bigr\}\). 
Figure~\ref{fig5} shows the graphs of the estimates of $\Re\F[\bar{\nu}]$ for different number of observations $n$. The quality of this estimate increases with the growth of $n$.

\begin{figure}[h]
    
    \centering
    
    % Ряд 1: два изображения
    \includegraphics[width=0.49\textwidth]{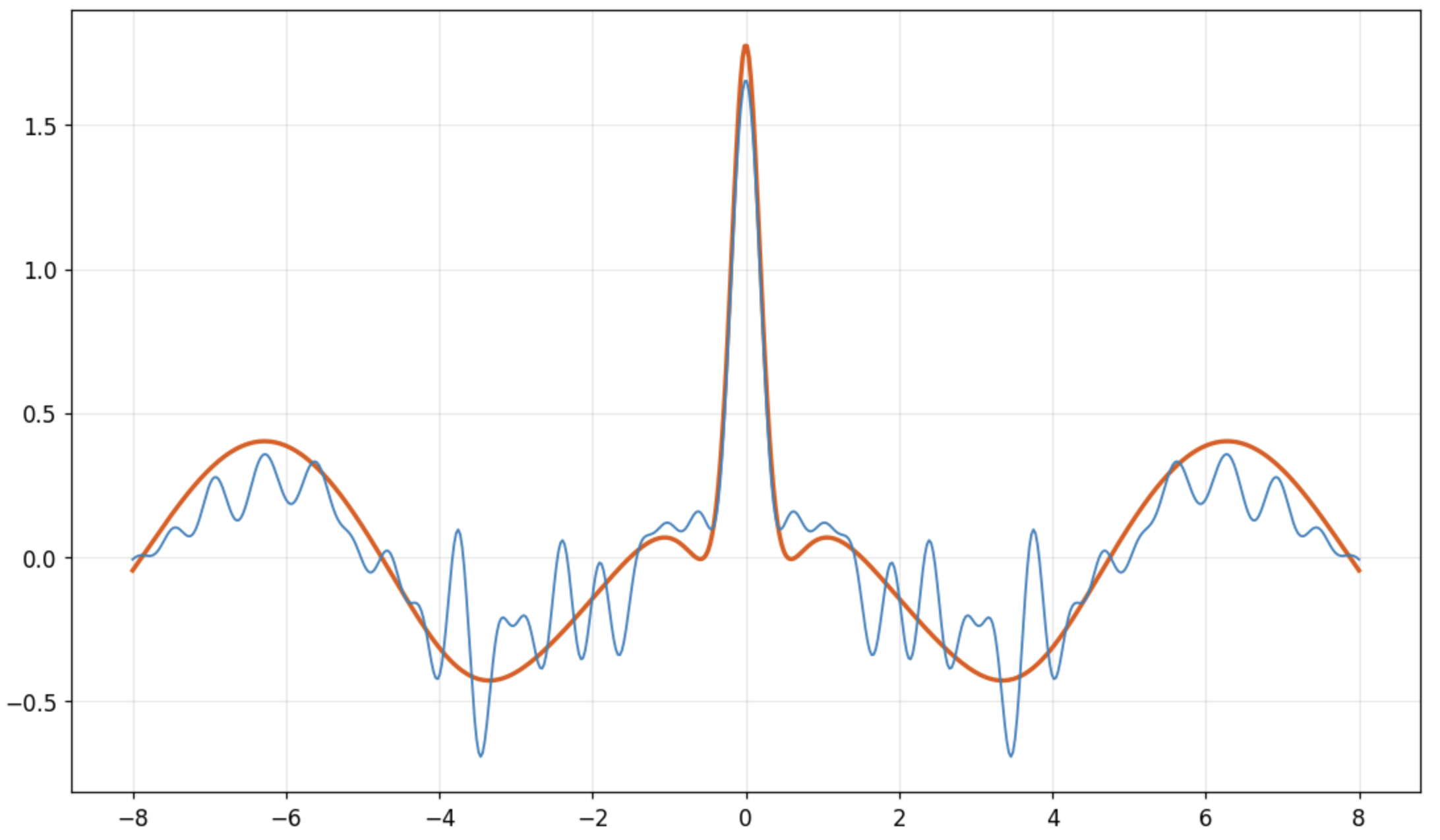}
    \hfill
    \includegraphics[width=0.49\textwidth]{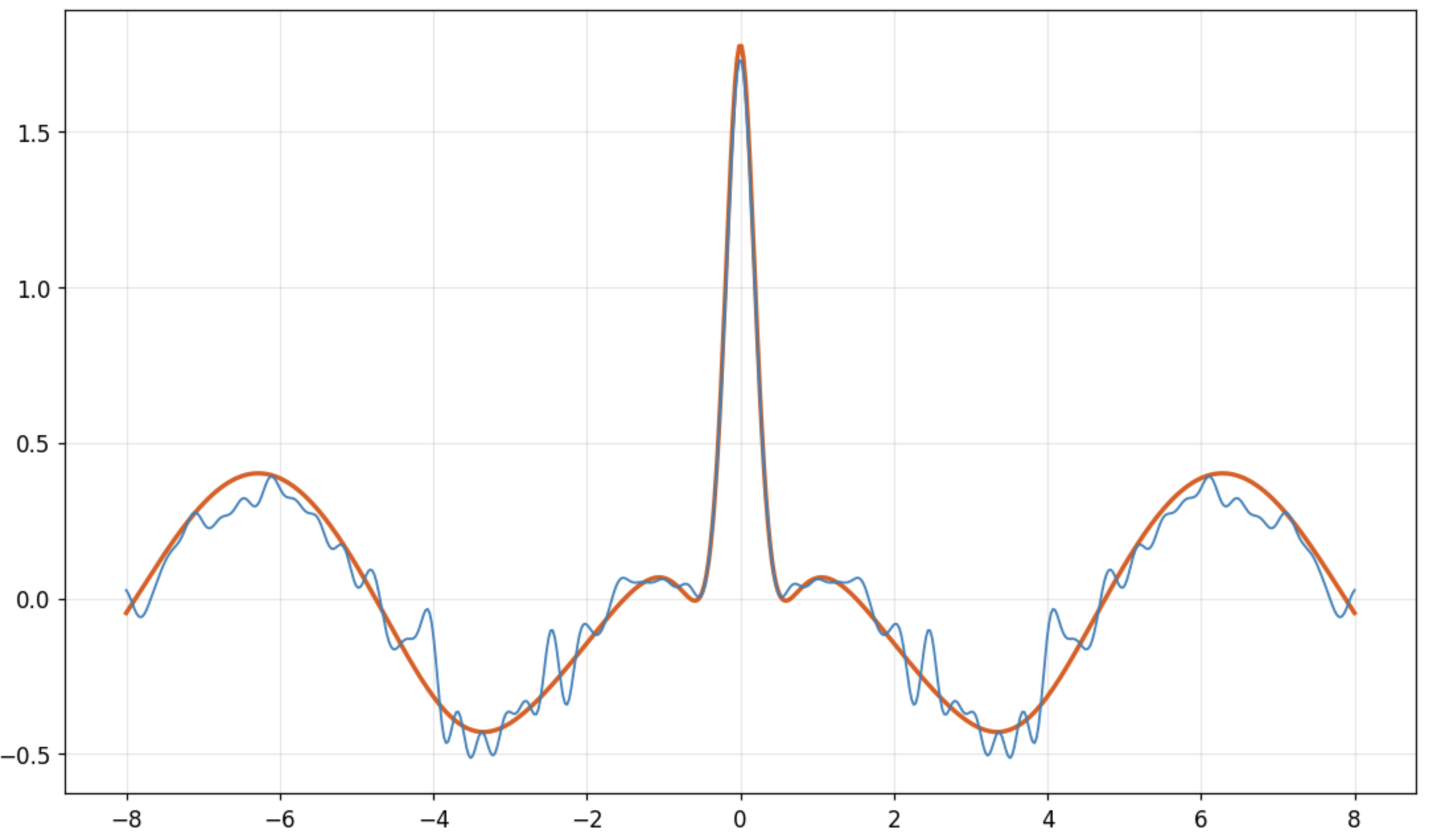}
    
    %  \vspace{1em}
    
    % % Ряд 2: одно изображение по центру
    % \includegraphics[width=0.45\textwidth]{images/nu_5000.png}
    
    \caption{\label{fig5} Plots of the true function $\Re\F[\bar{\nu}]$ (orange line) and its estimate $\Re(-\widehat{\psi}''(u) \F[K](u))$ (blue line) based on \(n = 5000\) and $n = 50000$ observations.}
\end{figure}

Finally, Figure~\ref{fig6} shows the graphs of the estimates of $\Re\F[\bar{\nu}_{Pois}]= \lambda_1\cos u$ for different number of observations $n$, whose quality also improves with the growth of $n$. \begin{figure}[h]
    
    \centering
    
    % Ряд 1: два изображения
    \includegraphics[width=0.49\textwidth]{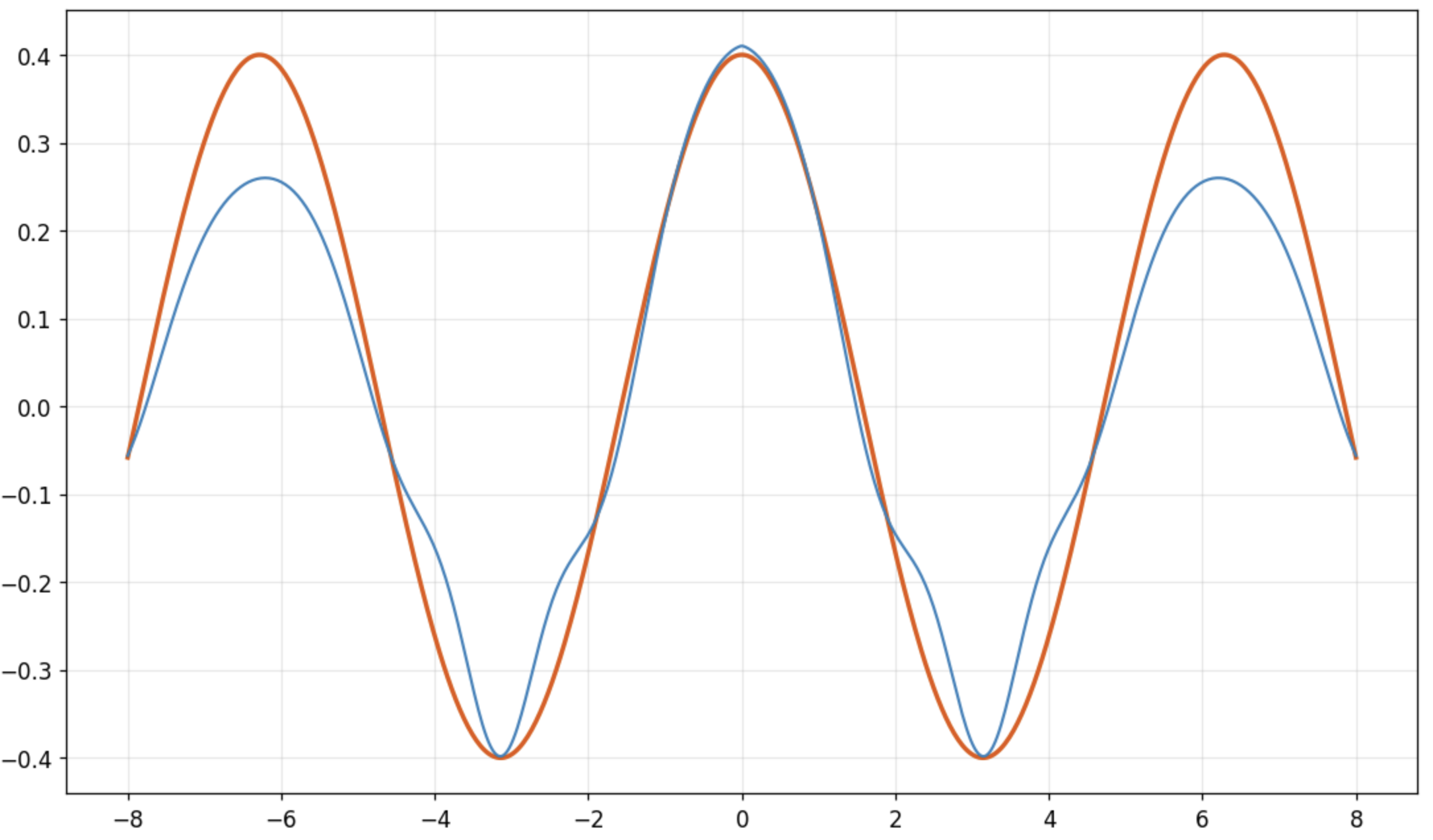}
    \hfill
    \includegraphics[width=0.49\textwidth]{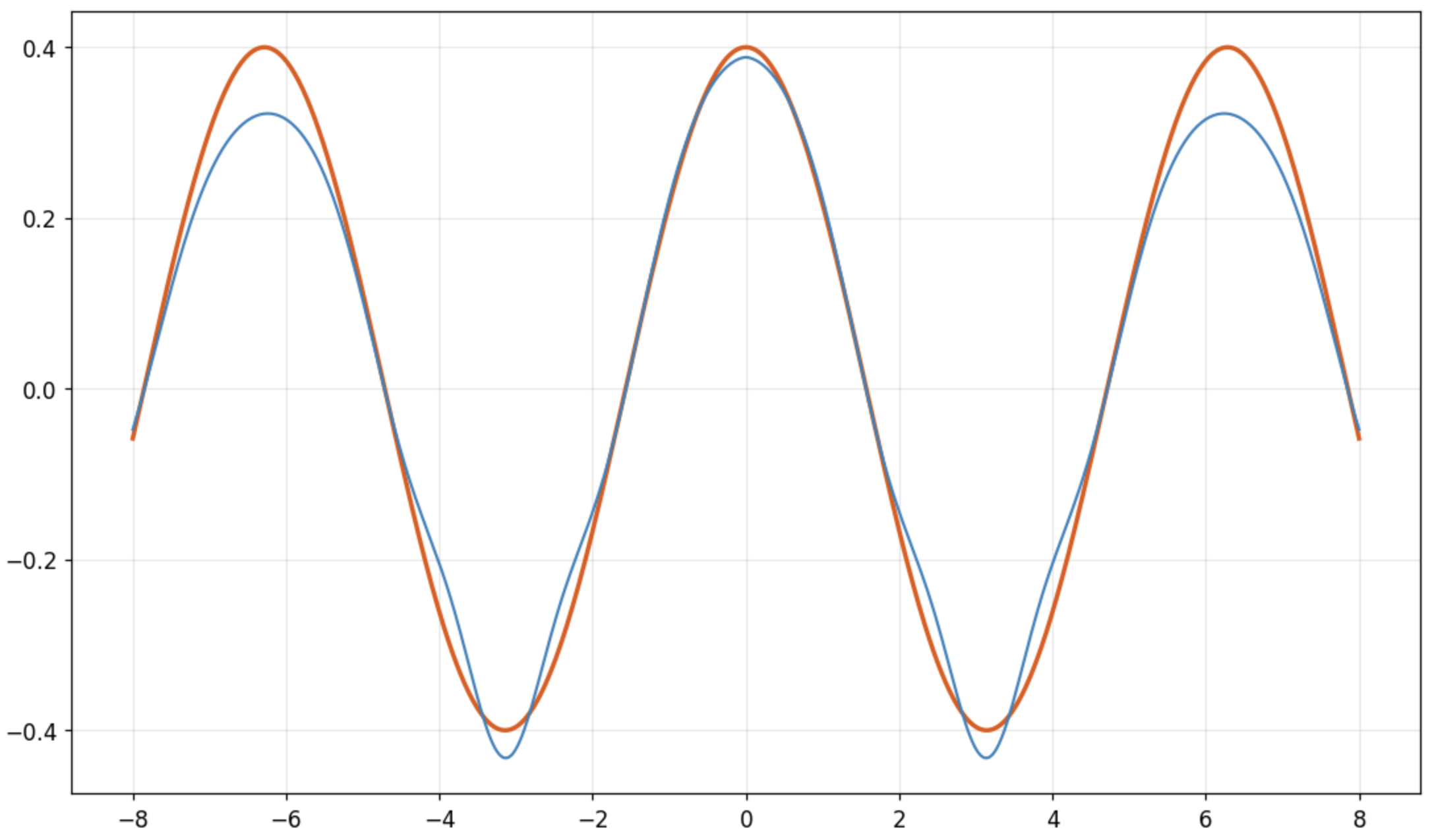}
    
    %  \vspace{1em}
    
    % % Ряд 2: одно изображение по центру
    % \includegraphics[width=0.45\textwidth]{images/nu_5000.png}
    
    \caption{\label{fig6} Plots of the true function $\Re\F[\bar{\nu}_{Pois}]$ (orange line) and its estimate $\Re(-\widehat{\psi}''_{Pois}(u) \F[K](u))$ (blue line) based on \(n = 5000\) and $n = 50000$ observations.}
\end{figure}

\section{Proofs}\label{proofs}

\subsection{Some results for empirical processes}
This section includes various results related to empirical characteristic processes, which will be helpful for proving the main theorems presented in the paper. For some interval~\(I\), denote a ``truncated'' version of the characteristic function 
\begin{align}\label{not1}
\phi_I(u) := \E\bigl[
e^{\i u X} \I\{X \in I\}
\bigr],
\end{align}
and its empirical counterpart, 
\begin{align}\label{not2}
\widehat\phi_I(u) := \frac{1}{n}\sum_{k=1}^n
e^{\i u X_k} \I\{X_k \in I\}.
\end{align}
The following lemma holds.
\begin{lemma}\label{our_lemma1}
For any unbounded sequence $U_n$ of positive numbers,\begin{align*}
      \max_{u \in [-U_n, U_n]}  |\widehat\phi_I(u) - \phi_I(u)| = O_{\P} \Bigl( 
        \sqrt{
\log(nU_n^2) / n
}
        \Bigr).
    \end{align*}
 \end{lemma}
\begin{proof} 
The proof of this fact is similar to the proof of Proposition 3.3 from \cite{BR2015}. Introduce random variables \(\eta_1,..,\eta_n\) that are i.i.d., centered and equal to 
\begin{align*}
\eta_k(u) &= e^{\i u X_k} \I\{X_k \in I\}
-\E\bigl[
e^{\i u X_k} \I\{X_k \in I\}
\bigr].
\end{align*}
Trivially, \(\widehat\phi_I(u) - \phi_I(u) = \bigl( \sum_{k=1}^n \eta_k(u)\bigr)/n=: S_n(u)/n.\)
 For \(R>0,\) denote the event 
\begin{equation*}
    \mathcal{A}_n = \mathcal{A}_n (R) := \Bigl\{\max_{u \in [-U_n,U_n]} \bigl| S_n(u) \bigr| \geq  R\sqrt{n \log(nU_n^2)}
     \Bigr\}, \qquad n=1,2,....
\end{equation*}
For the proof of this lemma, it is sufficient to show that \( \P\{\A_n\} \to 0\) as \(n \to \infty\) for some \(R>0.\)
%    For i.i.d. random variables $\eta_k$, $k = 1, \dots n$ with $\eta_k \in L^1$
% and any constants $R > 8$, $K \geq 1$ the following inequality holds for any $j \in J$
% \begin{equation*}
%     \P\bigl( \mathcal{A}_n (R) \bigr) \geq 1 -  C(\sqrt{n}K)^{(64-R^2)/128}
% \end{equation*}
% for some constant $C$ depending on $\E|X_1|$ only.
Note that
    \begin{align*}
        \P(\mathcal{A}_n) 
        &\leq   \P(\max_{u \in [-U_n, U_n]}  |\Re S_n(u)| \geq R\sqrt{n \log(nU_n^2)}/2)  + \P(\max_{u \in [-U_n, U_n]}|\Im S_n(u)| \geq R\sqrt{n \log(nU_n^2)}/2).
    \end{align*}
Below, we will consider in detail the real part. The proof for the imaginary part is similar.   Note that 
    \begin{align*}
        \Re S_n(u) &= \sum_{k=1}^n \Re \eta_k(u) = \sum_{k=1}^n \cos(uX_k)\I\{X_k \in I\} - \E [ \cos(uX_k)\I\{X_k \in I\}].
    \end{align*}
    Since $\Re \eta_k(u)$ $k = 1,\dots n$ are i.i.d. centered random variables, bounded by $2$, we can use Hoeffding\'s inequality (see Proposition 3.2 from \cite{BR2015}), which gives for any \(t>0\)
    \begin{equation*}
        \P(| \Re S_n(u)| \geq t/2) \leq 2\exp(-\frac{t^2}{32n}).
    \end{equation*}
    Next for some $J=J(n)$ we consider points $u_j = jU_n/J$ from an equidistant grid on $[-U_n, U_n]$, which leads to
    \begin{align*}
         \P(\max_{u_k}  |\Re S_n(u_k)| \geq t/2) &\leq \sum_{k = 1}^{2J}\P(| \Re S_n(u_k)| \geq t/2) \leq 4J \exp(-\frac{t^2}{32n}).
    \end{align*}
    Then, using  Lagrange's mean value theorem for arbitrary $u, v$, with $u \leq v$ we get
    \begin{equation*}
        |\cos(uX_k) - \cos(vX_k)| \leq |X_k||\sin(\theta X_k)||u - v| \leq |X_k||u - v|
    \end{equation*}
    for some $\theta \in [u, v]$.
    Hence, the following holds for any \(u, v \in \R,\)
    \begin{equation*}
        |\Re S_n(u) - \Re S_n(v) | \leq \sum_{k=1}^n \bigl(|X_k|\I\{X_k \in I\} + \E [|X_k|\I\{X_k \in I\}]\bigr)|u - v|.
    \end{equation*}
    Thus, we have
    \begin{align*}
         \P(\max_{u \in [-U_n, U_n]}  &|\Re S_n(u)| \geq t) \leq \P\Bigl(\max_{u_k} |\Re S_n(u_k)| +\sum_{k=1}^n \bigl(|X_k|\I\{X_k \in I\} + \E [|X_k|\I\{X_k \in I\}]\bigr) U_nJ^{-1}  \geq t\Bigr)\\
         &\leq \P\Bigl(\max_{u_k} |\Re S_n(u_k)| \geq t/2\Bigr) + \P\Bigl(\sum_{k=1}^n \bigl(|X_k|\I\{X_k \in I\} + \E [|X_k|\I\{X_k \in I\}]\bigr) U_nJ^{-1}  \geq t/2\Bigr).
    \end{align*}
By Markov's inequality we have
    \begin{equation*}
        \P(\max_{u \in [-U_n, U_n]}  |\Re S_n(u)| \geq t) \leq 4J\exp(-\frac{t^2}{32n}) + 4(nU_n/t)J^{-1}\E[|X_k|\I\{X_k \in I\}],
    \end{equation*}
    where \(\E[|X_k|\I\{X_k \in I\}] \leq \E|X_1| < \infty.\) Thus, the choice $J = \sqrt{nU_n/t}\exp{(t^2/(64n))}$ yields the order
    \begin{equation*}
        \P(\max_{u \in [-U_n, U_n]}  |\Re S_n(u)| \geq t) \leq 8 \sqrt{nU_n/t}\exp(-t^2/(64n)).
    \end{equation*}
To conclude the proof, we substitute \(t=R\sqrt{n\log(nU_n^2)}/2\), and get  
for $R > 8$ 
    \begin{equation*}
        \P\Bigl(\max_{u \in [-U_n, U_n]}  |\Re S_n(u)| \geq \frac{R}{2}\sqrt{n\log(nU_n^2)}\Bigr) \lesssim (\sqrt{n}U_n)^{(64-R^2)/128} \to 0, \qquad n \to \infty.
    \end{equation*}
\end{proof}

The next lemma yields the convergence rates for the derivatives of the empirical characteristic process.

\begin{lemma}\label{our_lemma2}
    Assume that the distribution $\mu$ has finite $4$-th moment with $\E|X_1|^4 \leq C_1$. Then
    \begin{equation*}
        \max_{u\in [-U_n, U_n]} |\widehat{\phi}^{(k)}(u) - \phi^{(k)}(u)| = O_{\mathbb{P}} \Bigl(C_1^{k/4}\Bigl(\frac{\log(nU^2_n)}{n}\Bigr)^{1/(k + 2)}\Bigr), \quad k=0,1,2.
    \end{equation*}
\end{lemma}
\begin{proof}
    Consider
\begin{equation*}
    \widehat{\phi}^{(k)}(u) - \phi^{(k)}(u) = \frac{1}{n}\sum_{j=1}^{n} \xi^{(k)}_j, \quad  \xi^{(k)}_j = (iX_j)^ke^{iuX_j} - \E (iX_j)^ke^{iuX_j}, \quad  k=0,1,2.
\end{equation*}
Now for each $k=0,1,2$ and  sequences $B_{n,k} \in \R$, which will be fixed later, we split $X_j$ into two parts:
\begin{equation*}
    X_j = X_j \I\{X_j \leq B_{n,k}\} + X_j \I\{X_j > B_{n,k}\} = X_{1, j} +X_{2, j}.
\end{equation*}
and then split $\xi^{(k)}_j$ accordingly into  $\xi^{(k)}_{1, j}$ and $\xi^{(k)}_{2, j}$.  Analogously to Lemma \ref{our_lemma1}, we get 
\begin{equation*}
    \max_{u \in [-U_n, U_n]} |\frac{1}{n}\sum_{j=1}^{n}\xi^{(k)}_{1, j}| =O_{\mathbb{P}} \Bigl(B_{n,k}^k \sqrt{\frac{\log(nU_n^2)}{n}}\Bigr).
\end{equation*}
Now we proceed to $\xi^{(k)}_{2, j}$. 
Using the Cauchy-Schwarz  and Markov inequalities, together with $\E|X_1|^{2k}\leq (\E|X_1|^4)^{k/2}\leq C_1^{k/2}$ for $k=0,1,2,$ we have
\begin{align*}
    \E |\frac{1}{n}\sum_{j=1}^{n}\xi^{(k)}_{2, j}| 
    &\leq
    2 \E \Bigl[ |X_1|^k\I\{|X_1| > B_{n,k}\} \Bigr] 
    \leq 2 \sqrt{\E |X_1|^{2k}} \sqrt{\mathbb{P}\{|X_1| > B_{n,k}\}}  \leq 2\, C_1^{(k+2)/4} B_{n,k}^{-2}.
\end{align*}
Then the choice $B_{n,k} = C_1^{1/4}\bigl(n / \log(nU_n^2) \bigr)^{1/(2k + 4)}$ balances the bounds for \(X_{1, j}\) and \(X_{2, j}\) and leads to the desired result. 
\end{proof}

\subsection{Proof of Theorem \ref{thm1}}
\begin{enumerate}
\item Since the theoretical value of \(\p_j, j \in \mJ,\) can be represented as
 \begin{align*}
\p_j = \w \bigl|
      (K_{c_j,\delta} * \phi_d)(u)       
\bigr|
=
\bigl|
      (K_{c_j,\delta} * \phi)(u)       
      -
      (1-\w)       (K_{c_j,\delta} * \phi_{ac})(u)
\bigr|
\end{align*}
for any \(u \in \R,\)  we have
\begin{align*}
\bigl| \widehat{\p}_j - \p_j \bigr| &=
\Bigl|
\int_{\R} w^{U_n} (u) 
\Bigl(
 \bigl| 
      (K_{c_j,\delta} * \widehat{\phi})(u)       \bigr|
-     \bigl|  (K_{c_j,\delta} * \phi)(u)       
      -
      (1-\w)       (K_{c_j,\delta} * \phi_{ac})(u)\bigr|
   \Bigr) du\Bigr|
      \\
      &\leq 
\int_{\R} w^{U_n} (u)  \Bigl|
 \bigl| 
      \bigl(K_{c_j,\delta} * \widehat{\phi}\bigr)(u)       \bigr| 
 -     \bigl|  (K_{c_j,\delta} * \phi)(u)       
      -
      (1-\w)       (K_{c_j,\delta} * \phi_{ac})(u)\bigr|
   \Bigr| du\\
     &\leq 
\int_{\R} w^{U_n} (u) 
\Bigl(
 \bigl| 
      \bigl(K_{c_j,\delta} * \bigl(\widehat{\phi}-\phi\bigr)\bigr)(u)       \bigr| 
+ (1-\w)  \bigl|     (K_{c_j,\delta} * \phi_{ac})(u)\bigr| \Bigr) du =: I_1+I_2.
\end{align*}
Let us consider separately the terms \(I_1\) and \(I_2\). Due to~\eqref{convolution}, we have 
\begin{align*}
I_1 &\leq 
\int_{\R} w^{U_n} (u) 
\bigl| 
      \bigl(K_{c_j,\delta} * \bigl(\widehat{\phi}-\phi\bigr)\bigr)(u)       \bigr| du = \int_{\R} w^{U_n} (u) |
      \widehat{\phi}_{-I_j} (u) - {\phi}_{-I_j}(u)
      |  
du,
\end{align*}
see the notations~\eqref{not1}-\eqref{not2}. Lemma \ref{our_lemma1} yields
\begin{align}\label{res1}
\max_{u \in [-U_n, U_n]} 
    \bigl| \bigl(K_{c_j,\delta} * \bigl(\widehat{\phi}-\phi\bigr)\bigr)(u)       \bigr|
     =
     O_{\P} \Bigl( 
        \sqrt{
\log(nU_n^2) / n
}
        \Bigr),
\end{align}
leading to \(I_1 =  O_{\P} \Bigl( 
        \sqrt{
\log(nU_n^2) / n
}
        \Bigr)\).
Now we turn towards the second term \(I_2.\) We have 
%\begin{eqnarray*}
%I_2 = (1-\w) 
%\int_\eps^1  \w(u) \bigl|     (K_{c_j,\delta} * \phi_{ac})(uU_n)\bigr| du
%\end{eqnarray*}
\begin{align}\nonumber
(K_{c_j,\delta} * \phi_{ac})(u) &= \F^{-1} \Bigl[ \F\bigl[
K_{c_j,\delta} * \phi_{ac}(\cdot)
\bigr] \Bigr] (u )
= \F^{-1} \Bigl[ \F\bigl[
K_{c_j,\delta} \bigr](\cdot) \F\bigl[ \phi_{ac}
\bigr](\cdot) \Bigr] (u)\\ \label{phiac}
&= 2\pi \F^{-1} \Bigl[  \I\{\cdot \in I_j\}  g_{ac}(-\cdot) \Bigr] (u) = \F\bigl[ g_{ac}(\cdot)  \I\{-\cdot\in I_j\} \bigr](u),
\end{align}
where we use that \(\F\bigl[ \phi_{ac}
\bigr](\cdot) = \int_{\R} \phi_{ac}(u) e^{\i u \cdot} du = 2\pi g_{ac}(-\cdot).\) Due to the Cauchy\,--\;Schwarz inequality and the Plancherel theorem, 
\begin{align}\nonumber
I_2 &\leq (1-\w) \Bigl(\int_{\R} \bigl(w^{U_n} (u) \bigr)^2 du\Bigr)^{1/2}
\Bigl(
\int_{\R} \bigl|
\F\bigl[ g_{ac}(\cdot)  \I\{-\cdot\in I_j\} \bigr](u)\bigr|^2 du
\Bigr)^{1/2}\\
\label{res2}
&\leq
U_n^{-1/2} (1-\w)
\Bigl(2 \int_{\eps}^1 \bigl(w (u) \bigr)^2 du\Bigr)^{1/2}
\sqrt{2 \pi}  \Bigl(
\int_{-I_j} 
\bigl( g_{ac}(x)  \bigr)^2 dx
\Bigr)^{1/2}.
\end{align}
Combining \eqref{res1} with \eqref{res2}, we get the  convergence rate~\eqref{ratepj} for $\widehat{p}_j$

\item Recall that \(\Arg(z)=\Im(\log(z)), z \in \C,\) where \(\log\) function is the continuous principal branch of the complex logarithm. Consider the mean absolute error for $x_j$, \(j \in \mJ \cap \widehat{\mJ},\)
\begin{align}\label{xj}
    |\widehat{x}_j - x_j| &\leq
   \int_{\R} | \widetilde{w}^{U_n} (u) | \cdot
| \Im G_j(u) |
 du,
%     \frac{\int_{\R} w^{U_n} (u) 
%\Im
%G_j(u) 
% du}{U_n \int_{\eps}^1 w (u) u  \, du},
\end{align}
where 
\begin{align}\label{Gj}
    G_j(u)  &=
 \log\bigl( K_{c_j,\delta} * \widehat{\phi}(u)\bigr)  -  \log \bigl(K_{c_j,\delta} * (\phi - (1-\w)\phi_{ac})(u)\bigr).
 \end{align} 
 The following lemma plays an important role. 
\begin{lemma}\label{lem5}
Let \(\phi\) be the characteristic function of a distribution from the class $\mu \in \mathcal{S}(\w_{\min}, \w_{\max}, C, \delta)$, and let \(\p_\circ>0\) be a fixed threshold parameter. Then for any \(j \in \mJ \cap \widehat{\mJ},\)
\begin{align}\label{lem51}
\inf_{|u| \in [\eps U_n, U_n]}|K_{c_j,\delta} * \phi(u)| \geq \p_\circ - r_j,
\end{align}
where 
\begin{align*}
r_j:=| \p_j  - \widehat{\p}_j| + \frac{2C}{\varepsilon U_n} \Bigl(1 +  \delta \Bigr)\to 0, \qquad \mbox{as}\quad n \to \infty,\end{align*}
while the characteristic function of the absolutely continuous part satisfies 
\begin{align}\label{lem52}
\sup_{|u| \in [\eps U_n, U_n]}|K_{c_j,\delta} * \phi_{ac}(u)|  \leq 
\frac{2C}{\varepsilon U_n} \Bigl(1 +  \delta \Bigr) \to 0,\qquad \mbox{as}\quad n \to \infty.
\end{align}
\end{lemma}
\begin{proof} 
We have
\begin{align*}
   |K_{c_j,\delta} * \phi(u)| &= \int_{\R} K_{c_j,\delta}(u - v)\phi(v)dv = \int_{\R} K_{c_j,\delta}(u - v)\Bigl( \int_\R e^{ivy} d \mu(y) \Bigr) dv \\ 
    &= \int_{\R} \int_{\R} K_{c_j,\delta}(u-v) e^{i(u - v)(-y)} e^{iuy} d \mu (y) dv= \int_{\R} \F[K_{c_j,\delta}](-y) e^{iuy} d \mu (y)  \\
    &= \w \sum_{k} \F[K_{c_j,\delta}](-x_k) e^{\i u x_k} p_k + (1 - \w) \int_{-I_j}  e^{iuy} d \mu_{ac} (y) \\
    &= \p_j \e^{\i u x_j} + (1 - \w) K_{c_j,\delta} * \phi_{ac}(u),%, \qquad \mbox{where} \; l_{ac}(u):= \int_{-I_j}  e^{iuy} d \mu_{ac} (y),
\end{align*}
due to the Fubini theorem and \eqref{eq37}. Therefore,  for any \(u \in \R,\)
\begin{eqnarray*}
|K_{c_j,\delta} * \phi(u)|  \geq \p_j - (1-\w) |K_{c_j,\delta} * \phi_{ac}(u)|,
\end{eqnarray*}
and 
\begin{align*}
\inf_{u \in [\eps U_n, U_n]}|K_{c_j,\delta} * \phi(u)| &\geq \p_j  - (1-\w) \sup_{u \in [\eps U_n, U_n]}|K_{c_j,\delta} * \phi_{ac}(u))|,
\end{align*}
where \(\p_j \geq \widehat{\p}_j - | \p_j  - \widehat{\p}_j| \geq \p_\circ-| \p_j  - \widehat{\p}_j|.\)
Finally, for any \(u \in [\eps U_n, U_n],\)  we continue the line of reasoning in~\eqref{phiac},
\begin{align*}
    |K_{c_j,\delta} * \phi_{ac}(u)| &= |\int_{\R} \I\{ -y \in I_j \} g_{ac}(y)e^{iuy}dy| = |\int_{I_j} g_{ac}(-y) d \frac{e^{-iuy}}{-iu}|  \\
    &\leq  \frac{1}{|u|} \Bigl(|g_{ac}(-c_j - \delta)| + |g_{ac}(-c_j + \delta)| + \int_{I_j} |g'_{ac}(-y)| dy \Bigr)\leq \frac{2C}{\varepsilon U_n} (1 +  \delta ),\end{align*}
and arrive at the desired result.
\end{proof}
From Lemma~\ref{lem5} it follows that \(K_{c_j,\delta} * \phi(u)  \ne 0\) for all \(u\) with  \(|u| \in [\eps U_n, U_n]\) and \(n\) large enough, and \(|K_{c_j,\delta} * \phi_{ac}(u)| \to 0\) as \(n \to \infty\).  Continuing the line of reasoning in~\eqref{Gj}, we get 
 \begin{align*}
   G_j(u)  &=
     \log \bigl(K_{c_j,\delta} * \widehat{\phi}(u)\bigr)  - \log \bigl(K_{c_j,\delta} * \phi(u) \bigr) - \log\Bigl(1 - (1-\w )\frac{K_{c_j,\delta} * \phi_{ac}(u)}{K_{c_j,\delta} * \phi(u)}\Bigr) \\
    &=\log \Bigl( 1 + \frac{K_{c_j,\delta} * (\widehat{\phi} - \phi)(u)}{K_{c_j,\delta} * \phi(u)}\Bigr)- \log\Bigl(1 - (1- \w )\frac{K_{c_j,\delta} * \phi_{ac}(u)}{K_{c_j,\delta} * \phi(u)}\Bigr),
\end{align*}
and therefore 
\begin{align*}
    |\widehat{x}_j - x_j| &\leq
 \int_{\R}  |\widetilde{w}^{U_n} (u) |
\Bigl| \Im \log \Bigl( 1 + \frac{K_{c_j,\delta} * (\widehat{\phi} - \phi)(u)}{K_{c_j,\delta} * \phi(u)}\Bigr) \Bigr|
 du \\ 
 & + 
 \int_{\R} |\widetilde{w}^{U_n} (u)| 
\Bigl| \Im \log\Bigl(1 - (1- \w )\frac{K_{c_j,\delta} * \phi_{ac}(u)}{K_{c_j,\delta} * \phi(u)}\Bigr) \Bigr|
 du =: I_3+I_4.
%     \frac{\int_{\R} w^{U_n} (u) 
%\Im
%G_j(u) 
% du}{U_n \int_{\eps}^1 w (u) u  \, du},
\end{align*}
In what follows, we will  consider $I_3$ and $I_4$ separately. From  the inequality \begin{align}\label{main}|\Im\log(1+z)|\leq|\log(1+z)|\leq 2|z|  \qquad \text{for} \; |z| < 1/2,\end{align} we get 
\begin{align*}
   I_3 & \leq 
   2 
      \int_{\R} |\widetilde{w}^{U_n} (u)| \bigl| 
   \frac{K_{c_j,\delta} * (\widehat{\phi} - \phi)(u)}{K_{c_j,\delta} * \phi(u)}
   \bigr|du
   \\
    &\leq 
    \frac{2\int_\R |\widetilde{w}(v)|dv}{U_n}
 \frac{\sup_{|u| \in [ \eps U_n, U_n]} \bigl| \bigl(K_{c_j,\delta} * \bigl(\widehat{\phi}-\phi\bigr)\bigr)(u)       \bigr|}{\inf_{|u| \in [\eps U_n, U_n]}|K_{c_j,\delta} * \phi(u)|}
 = 
     O_{\P} \Bigl( 
        \frac{
                \sqrt{
\log(nU_n^2)/(n) } }{ U_n \p_\circ}
        \Bigr),
\end{align*}
see~\eqref{res1} and \eqref{lem51}. For \(I_4\), we apply the same inequality~\eqref{main}, and analogously to \eqref{phiac}-\eqref{res2} get 
\begin{align*}
     I_4 & \leq 
   2 (1-\w)
      \int_{\R} |\widetilde{w}^{U_n} (u)| \bigl|
      \frac{K_{c_j,\delta} * \phi_{ac}(u)}{K_{c_j,\delta} * \phi(u)}
      \bigr| du \\
      &\lesssim
 2U_n^{-3/2} (1-\w)
\Bigl(2 \int_{\eps}^1 \bigl(\widetilde{w} (u) \bigr)^2 du\Bigr)^{1/2}
\sqrt{2 \pi}  \Bigl(
\int_{-I_j} 
\bigl( g_{ac}(x)  \bigr)^2 dx
\Bigr)^{1/2} \p_{\circ
}^{-1}\lesssim \frac{\sqrt{\delta}C}{U_n^{3/2} \p_{\circ
}}.
\end{align*}
The obtained bounds for \(I_3\) and \(I_4\) lead to  the statement of the theorem.
\item We proceed with the convergence rates for $\widehat\w$. We have
\begin{align}
    |\widehat{\w} - \w| &= |\sum_{j \in \widehat{\mJ}} \widehat{\p}_j - \sum_{j=1}^J  \p_j| \leq \sum_{j =1}^J |\widehat{\p}_j - \p_j| +     \sum_{j \notin \widehat\mJ} |\widehat{\p}_j| \leq  \sum_{j =1}^J |\widehat{\p}_j - \p_j| + c J U_n^{-1} = O_{\P}(\mQ_n), \label{www2}
\end{align}
since \(U_n^{-1} =O(\mQ_n).\)
\item Finally, we consider the estimation error for $p_j$ for fixed $j$, 
\begin{align*}
    |\widehat{p}_j - p_j| &= |\frac{\widehat{\p}_j}{\widehat{\w}} - \frac{\p_j}{\w}| 
    \leq \frac{| \widehat{\p}_j - \p_j |}{\widehat{\w}} + \p_j\frac{ |\widehat{\w} - \w|}{\widehat{\w}\,\w}
    \leq 
    \frac{1}{\widehat{\w}} \bigl( 
    | \widehat{\p}_j - \p_j | + 
    \frac{ |\widehat{\w} - \w|}{\w_{\min}}
    \bigr).
\end{align*}
Application of the inequality \(\widehat{\w} \geq \w - | \widehat{\w}   -\w| \geq \w_{\min}-| \widehat{\w}   -\w|\) leads to the desired result. 
\end{enumerate}
\subsection{Proof of Theorem \ref{thm2}}

\begin{proof}

We start from the identity
\begin{align*}
 \widehat{g}_{ac}(x) -  g_{ac}(x) &= \frac{1}{2 \pi}\Bigl( \int_{|u| \leq V_n} \Bigl( \widehat{\phi}_{ac}(u) - \phi_{ac}(u)  \Bigr) e^{-\i u x}  du 
 - \int_{|u| >V_n} \phi_{ac}(u)  e^{-\i u x}  du\Bigr):=\frac{1}{2 \pi}\Bigl(I_1 + I_2\Bigr).
\end{align*}
For the first summand, consider the representation 
\begin{align*}
    \widehat{\phi}_{ac}(u) - \phi_{ac}(u) 
   &= \frac{\widehat{\phi}(u) - \widehat{\w} \widehat{\phi}_d(u)}{1 - \widehat{\w}} - \frac{\phi(u) - \w \phi_d(u)}{1 - \w} \nonumber\\
       &= \frac{1}{1-\widehat\w} \bigl( \widehat{\phi}(u) - \phi(u) \bigr) 
       - \frac{1}{1-\widehat\w} \bigl( \widehat{\w} \widehat{\phi}_d(u) - \w \phi_d(u) \bigr) + \frac{\widehat{\w} - \w}{(1-\widehat{\w})(1-\w)} \bigl( \phi (u) - \w \phi_d(u) \bigr). 
%    &= \frac{1}{1-\w} \bigl( \widehat{\phi}(u) - \phi(u) \bigr) 
%       - \frac{1}{1-\w} \bigl( \widehat{\w} \widehat{\phi}_d(u) - \w \phi_d(u) \bigr) + \frac{\widehat{\w} - \w}{(1-\widehat{\w})(1-\w)} \bigl( \widehat{\phi}(u) - \widehat{\w} \widehat{\phi}_d(u) \bigr). 
\end{align*}
Since \(1-\widehat{\w} \geq 1 - \w-|\widehat{\w} - \w| \gtrsim 1 - \w_{\max},\) we get the pointwise bound
\begin{align}\label{bound_main}
    |\widehat{\phi}_{ac}(u) - \phi_{ac}(u)| 
    &\le C_1 |\widehat{\phi}(u) - \phi(u)| 
       + C_2 |\widehat{\w} \widehat{\phi}_d(u) - \w \phi_d(u)| 
       + C_3 |\widehat{\w} - \w| |\phi_{ac}(u)|, 
\end{align}
where $C_1, C_2, C_3>0.$ Let us consider the summands separately for \(u \in [-U_n, U_n]\). The bound for the first summand is given by Lemma~\ref{our_lemma1},
\begin{align*}
\sup_{u \in [-V_n, V_n]} |\widehat{\phi}(u) - \phi(u)| = O_{\P} \Bigl( 
        \sqrt{
\log(n V_n^2) / n
}
        \Bigr).
\end{align*}
For the second term in~\eqref{bound_main}, we decompose
\begin{align}\nonumber
|
    \widehat{\w} \widehat{\phi}_d(u) - \w \phi_d(u) 
    |
    &= |\widehat{\w} (\widehat{\phi}_d(u) - \phi_d(u)) + (\widehat{\w} - \w) \phi_d(u)|\\
     &\leq \bigl(\w_{\max}+O_{\P}(\mQ_n)\bigr)|\widehat{\phi}_d(u) - \phi_d(u)| + O_{\P}(\mQ_n),\label{bound_discrete}
\end{align}
where we use that $|\phi_d(u)| \le 1$ and $|\widehat{\w}| \leq \w_{\max} + |\widehat{\w} - \w| \leq \w_{\max}+O_{\P}(\mQ_n)$.

Now we estimate $|\widehat{\phi}_d(u) - \phi_d(u)|$. We have
\begin{align}
    |\widehat{\phi}_d(u) &- \phi_d(u)| 
    = \Bigl| \sum_{j \in \widehat{\mJ}}   \widehat{p}_j e^{iu\widehat{x}_j} -  \sum_{j =1}^J  p_j e^{iux_j}  \Bigr| 
    \leq
    \sum_{j \in \widehat{\mJ}}
    \Bigl|  \widehat{p}_j e^{iu\widehat{x}_j} -     p_j e^{iux_j}  \Bigr|  +  \sum_{j \notin \widehat{\mJ}} |p_j e^{iu x_j}| \label{bound_phi_d}\\
    &\le \sum_{j \in  \widehat\mJ}\Bigl( |\widehat{p}_j - p_j| |e^{iu\widehat{x}_j}| 
           + p_j |e^{iu\widehat{x}_j} - e^{iux_j}| \Bigr) 
           +
           \sum_{j \notin  \widehat\mJ}
           \bigl( \widehat{p}_j + |\widehat{p}_j - p_j|
           \bigr) = O_{\P}(\mQ_n) + \sum_{j=1}^{J} |\e^{iu(\widehat{x}_j - x_j)} - 1| \nonumber,
\end{align}
where, similar to~\eqref{www2}, we use that \( \sum_{j \in \widehat{\mJ}}
      \widehat{p}_j \leq c J U_n^{-1} =O(\mQ_n).\)
The bound $|e^{i\theta} - 1| \le |\theta|, \;\forall \theta \in \R$,
yields 
\(
|e^{iu(\widehat{x}_j - x_j)} - 1| \le |u| |\widehat{x}_j - x_j|,
\)
and we arrive at 
\begin{align}\label{615}
\sup_{u \in [-V_n, V_n]}     |\widehat{\phi}_d(u) - \phi_d(u)| = O_{\P}( V_n \mQ_n). 
\end{align}
Combining this result with~\eqref{bound_discrete}, we get 
\begin{align*}
\sup_{u \in [-V_n, V_n]}     | \widehat{\w} \widehat{\phi}_d(u) - \w \phi_d(u) | = O_{\P}( V_n \mQ_n). 
\end{align*}
To sum up, we have 
\begin{align*}
|I_1| &=O_\P \Bigl( 
V_n    \sqrt{
\log(n V_n^2) / n
}
 +V_n^2\mQ_n+V_n\,\mQ_n
\Bigr).
\end{align*}
As for \(I_2,\) we trivially have \(|I_2| \leq \int_{|u|>V_n} |\phi_{ac}(u)|du\). This observation concludes the proof of~\eqref{rateG}. For the particular cases~\eqref{rateP} and \eqref{rateE} we have 
\begin{align*}
\int_{|u|>V_n} |\phi_{ac}(u)|du \leq \int_{|u| > V_n} \mathtt{C} (1+|u|)^{-\alpha} du =
\frac{2\mathtt{C}}{(\alpha - 1)(1+V_n)^{\alpha-1}}\lesssim \frac{1}{V_n^{\alpha-1}}
\end{align*}
if \(g_{ac} \in \mathcal{P}_{\alpha, \mathtt{C}}\) with \(\alpha>2\), $\mathtt{C} > 0$. In the exponential case we get 
\begin{align*}
\int_{|u|>V_n} |\phi_{ac}(u)|du \leq \int_{|u|>V_n} e^{-\mathtt{C} u^2} du  = 2(\pi/ \mathtt{C})^{1/2} \Psi(\sqrt{2\mathtt{C}} V_n) \lesssim \frac{\e^{-\mathtt{C}V^2_n}}{V_n},
\end{align*}
where \(\Psi\) is the survival function of the standard normal distribution, if \(g_{ac} \in \mathcal{E}_{\gamma,\mathtt C}\) with \(\gamma>2\), \(\mathtt{C}>1.\)
\end{proof}

\subsection{Proof of Theorem \ref{thm3}}

The estimation error can be decomposed as follows: \begin{align*}
    \widehat{\bar{\nu}} - \bar{\nu} &= -\F^{-1}[\widehat{\psi}''\F[K_n]] - \bar{\nu} = -\F^{-1}[(\widehat{\psi}'' - \psi'')\F[K_n]] + \bigl( - \F^{-1}[\psi''\F[K_n]]   - \bar{\nu} \bigr)\\
        &= -\F^{-1}[\widehat{\psi}'' - \psi''] * K_n + \bigl(  K_n * \bar{\nu}   - \bar{\nu} \bigr)=: -I_1+ I_2,
\end{align*}
where we applied~\eqref{imp_eq} and the properties of the Fourier transform.
\begin{enumerate}
\item First, we consider the term $I_2$.  Its \(H^{-1}\)-norm is equal to 
    \begin{align*}
        \sup_{||f||_{H^1} =1} |\int_{\R} f(x) \,\, d(K_n * \bar{\nu} - \bar{\nu})| &= \sup_{||f||_{H^1} =1} |\int_{\R} \Bigl(\int_{\R} f(x) K_n(x-y) dx \Bigr)\, \bar{\nu}(dy) - \int_{\R} f(x) \bar{\nu}(dx)| \\
         &= \sup_{||f||_{H^1} =1} |\int_{\R}\int_{\R} f(x+y) K_n(x) dx \, \bar{\nu}(dy) - \int_{\R} f(x) \bar{\nu}(dx)|,
        \end{align*}
by the definition of the convolution of a function and a measure. Therefore, 
    \begin{align*}
        ||I_2||_{H^{-1}} &= \sup_{||f||_{H^1} =1} |\int_{\R} (f * K_n(-\cdot) - f)(x) \,\, \bar{\nu}(dx)| \leq \sup_{||f||_{H^1} =1} \sup_{x\in \R} |(f * K_n(-\cdot) - f)(x)| \,\, \int_{\R} |\bar{\nu}|(dx) \\
    &= \sup_{||f||_{H^1} =1} \sup_{x\in \R} |\int_{\R} (f(x + y) - f(x)) K_n(y)dy|\cdot |\bar{\nu}|(\R),
    \end{align*}
where we used that $\int K_n(x)dx =\int K(x)dx = 1$. Thus, using the Newton-Leibniz theorem and then the Cauchy-Schwarz inequality we get
    \begin{align*}
     ||I_2||_{H^{-1}} &\leq \sup_{||f'|| \leq 1} \sup_{x\in \R} |\int_{\R} (\int_x^{x+y} f'(z) dz)K_n(y)dy|\cdot |\bar{\nu}|(\R) \leq |\bar{\nu}|(\R) \sup_{||f'|| \leq 1} \int_{\R} |y|^{1/2} \cdot ||f'|| \cdot K_n(y)dy \\ 
    &\leq |\bar{\nu}|(\R)\int |y|^{1/2} K_n(y) dy
    = 
    |\bar{\nu}|(\R) W_n^{-1/2}\int  |z|^{1/2} K(z) dz\leq C_2\, W_n^{-1/2}\int  |z|^{1/2} K(z) dz \lesssim C_2\, W_n^{-1/2},
\end{align*}
where in the last step we used $|\bar\nu|(\R)\leq C_2$ for $\mu\in\mathcal{M}(C_1,C_2)$.

\item Now we proceed to the estimation of the term $I_1$. Recall that the \(H^{-1}\)-norm can be represented as
\begin{align}\label{H-1}
    ||I_1||^2_{H^{-1}} &= ||\F^{-1}[(\widehat{\psi}'' - \psi'')\F[K_n]]||^2_{H^{-1}} 
    = \frac{1}{2\pi} \int_{\R} \frac{|\widehat{\psi}''(u) - \psi''(u)|^2 |\F[K_n](u)|^2}{1 + u^2} du,
    \end{align}
  see~\eqref{first_norm}.   Denote \begin{equation*}
    \Delta_n := \frac{\widehat{\phi}(u) - \phi(u)}{\phi(u)}.
\end{equation*}
Lemma~5.1 in \citep{PanovRyabchenko2026} states that the probability of the event 
\begin{equation*}
    \mathcal{A}_n := \Bigl\{\max_{u \in [-W_n,W_n]} |\Delta_n| \leq \chi_n \Bigr\}, \qquad n=1,2,...,
\end{equation*}
with \[\chi_n := \chi^\circ \frac{\sqrt{\log(nW_n^2)/ n}}{\inf_{u \in [-W_n, W_n]}|\phi(u)|}, \qquad \chi^\circ<1/16,\]
tends to 1 as \(n \to \infty.\) More precisely, it is known that 
 \(\P\{\A_n\} \geq 1-c (\sqrt{n}W_n)^{-\kappa}\) with \(\kappa= (1/(2\chi^\circ)^2-64) / 128>0\) and some positive constant \(c,\) which  depends on \(\E[|X_1|]\) only. As we have discussed in the introduction, \(\mu \in \Q\) yields $\inf_{\R}|\phi(u)|  >0$ (see Theorem 2.2 from \cite{BK2023} for the proof). Therefore, \(\max_{u \in [-W_n,W_n]} |\Delta_n| = O_P\bigl(\sqrt{\log(nW_n^2)/ n}\bigr).\) Note also that the condition \(W_n \lesssim n\) guarantees that \(\chi_n \to 0\) as \(n \to \infty.\)
 
On the event $\A_n$ we have for any fixed $u$ 
\begin{align*}
    \widehat\psi(u) - \psi(u) = \log(1 + \Delta_n(u))
    = \Delta_n(u)  - \frac{\Delta_n^2(u)}{2} + O(|\Delta_n^3|) = \Delta_n(u) + R_n(u).
\end{align*}
This leads to the equality 
\(
    \widehat{\psi}''(u) - \psi''(u) = \Delta''_n(u) + R''_n(u),\)
where the term $R''_n(u)$ is equal to 
\begin{equation*}
    R''_n(u) = \bigl(    \log(1 + \Delta_n(u)) - \Delta_n(u) \bigr)'' = 
    -\frac{(\Delta'_n(u))^2}{(1 + \Delta_n(u))^2} - \frac{\Delta_n(u)\Delta''_n(u)}{1 + \Delta_n(u)}.
\end{equation*}
Since on the event $\A_n$ for $n$ large enough $\max_{u \in [-W_n, W_n]} |1 + \Delta_n(u)| \geq 1 - \chi_n > 0$, the  denominators of both fractions are bounded. Next, denote $\varepsilon (u) := \widehat{\phi}(u) - \phi(u), \; u \in \R,$ and consider the derivatives of $\Delta_n$,
\begin{align*}
    \Delta_n(u) &=  \frac{\varepsilon(u)}{\phi(u)}, \qquad 
    \Delta_n'(u) = \frac{\varepsilon'(u) \phi(u) - \phi'(u)\varepsilon(u)}{\phi^2(u)}, \\
    \Delta_n''(u) &= \frac{\varepsilon''(u)\phi^2(u)-\varepsilon(u) (\phi''(u)\phi(u) + 2\phi'(u)^2) - 2\varepsilon'(u)\phi'(u) \phi(u)}{\phi^3(u)}.
\end{align*}
Again, since  $\inf_{\R}|\phi(u)| >0$, the denominators of all fractions are separated from zero uniformly. Moreover, for $\mu\in\mathcal{M}(C_1,C_2)$ we have $|\phi^{(j)}(u)|\leq\E|X|^j\leq C_1^{j/4}$ for $j=0,1,2.$ %Same is true for all $\phi^{(k)}, \, k=0, 1, 2$, because
%\begin{equation*}
%    |\phi^{(k)}(u)| \leq \E |X^k e^{iuX_k}| = \E |X|^k < \infty.
%\end{equation*}
%Therefore, the behavior of $R''_n(u)$ depends on the behavior of $\varepsilon^{(k)}, \, k=0,1,2$. Same is true for  behavior of $\widehat{\psi}''(u) - \psi''(u)$.
Lemma \ref{our_lemma2} gives the following convergence rates for $\varepsilon^{(k)}, \, k=0,1,2$
\begin{equation*}
    \max_{u\in [-W_n, W_n]} |\varepsilon^{(k)}(u)| = \max_{u\in [-W_n, W_n]} |\phi^{(k)}_n(u) - \phi^{(k)}(u)| = O_{\mathbb{P}} \Bigl(C_1^{k/4}\Bigl(\frac{\log(nW^2_n)}{n}\Bigr)^{1/(k + 2)}\Bigr), 
\end{equation*}
which, together with $|\phi^{(j)}(u)|\leq C_1^{j/4},$ lead to the convergence rates for $\Delta''_n$ and $R''_n$
\begin{equation*}
        \max_{u \in [-W_n, W_n]}|\Delta''_n (u)| = O_{\mathbb{P}}\Bigl(C_1^{1/2}\Bigl(\frac{\log(nW^2_n)}{n}\Bigr)^{1/4}\Bigr), \quad     \max_{u \in [-W_n, W_n]}|R''_n(u)| = O_{\mathbb{P}}\Bigl(C_1^{1/2}\Bigl(\frac{\log(nW^2_n)}{n}\Bigr)^{2/3}\Bigr).
\end{equation*}
% \begin{eqnarray*}
% \Delta_n &=& O_{\mathbb{P}}\Bigl(\Bigl(\frac{\log(nU_n)}{n}\Bigr)^{1/2}\Bigr), \\
% \Delta'_n &=& O_{\mathbb{P}}\Bigl(\Bigl(\frac{\log(nU_n)}{n}\Bigr)^{1/3}\Bigr),\\
%     \Delta''_n &=& O_{\mathbb{P}}\Bigl(\Bigl(\frac{\log(nU_n)}{n}\Bigr)^{1/4}\Bigr),
% \end{eqnarray*}
% and therefore, $R''_n = O_{\mathbb{P}}\Bigl(\Bigl(\frac{\log(nU_n)}{n}\Bigr)^{2/3}\Bigr)$.
% Further we will use the result of the following Lemma, which gives the convergence rates for $R''_n$ and specifically for $\Delta''_n$, its proof can be found in the Appendix section.
Therefore, the estimation error of $\widehat{\psi}''(u) - \psi''(u)$ is dominated by the term $\Delta''_n(u) $. Continuing the line of reasoning in~\eqref{H-1}, we get 
    \begin{align*}    ||I_1||^2_{H^{-1}}
    &\lesssim \frac{1}{2\pi} \int_{-W_n}^{W_n} (1 + u^2)^{-1} |\Delta_n''(u) |^2 |\F[K](u/W_n)|^2 du 
    \lesssim O_{\mathbb{P}}\Bigl(\mG_n\, C_1\Bigl(\frac{\log(nW^2_n)}{n}\Bigr)^{1/2}\Bigr) ,
    \end{align*}
    where 
    \begin{align*} 
    \mG_n &:=   \int_{-1}^{1} (1 + W_n^2v^2)^{-1} |\F[K](v)|^2 W_n dv  \lesssim  \int_{-1}^{1} \frac{d W_n v}{1 + W_n^2v^2} = 2\,\arctan(W_n) \leq \pi.
\end{align*}
Combining the results for $I_1$ and $I_2$ we get the convergence rate for $\widehat{\bar{\nu}}$, which is
\begin{equation*}
    ||\widehat{\bar{\nu}} - \bar{\nu}||_{H^{-1}} = O_{\mathbb{P}}\Bigl(C_2\,\frac{1}{\sqrt{W_n}} + C_1^{1/2}\Bigl(\frac{\log (nW^2_n)}{n}\Bigr)^{1/4}\Bigr).
\end{equation*}
\end{enumerate}

\subsection{Proof of Theorem \ref{thm4}}

\begin{proof}
Similarly to the proof of Theorem \ref{thm3}, consider the decomposition 
\[
\widehat{\bar\nu}_d-\bar\nu_d
=-\F^{-1}\bigl[(\widehat\psi_d''-\psi_d'')\F[K_n]\bigr]
+\bigl(K_n*\bar\nu_d-\bar\nu_d\bigr)=:-I_1^d+I_2^d.
\]
The term $I_2^d$ is analyzed exactly as in Theorem \ref{thm3}. By Proposition~\ref{propmain}(i), $\bar\nu_d$ is the atomic part of $\bar\nu,$ so $|\bar\nu_d|(\R)\leq|\bar\nu|(\R)\leq C_2$ for $\mu\in\mathcal{M}(C_1,C_2),$ and
\[
\|I_2^d\|_{H^{-1}}\le|\bar\nu_d|(\R)\,W_n^{-1/2}\!\int|z|^{1/2}K(z)\,dz\lesssim C_2\, W_n^{-1/2}.
\]
Now we proceed to the estimation of the term $I_1^d$. 
First, we rewrite the point error of \(\widehat\phi_d^{(m)}(u)\) for fixed $m=0,1,2$ and $|u| \leq W_n$ as 
\[
\widehat\phi_d^{(m)}(u)-\phi_d^{(m)}(u)=\sum_{j\in\widehat{\mJ}}i^m\Bigl\{
\bigl(\widehat{p}_j\widehat{x}_j^m-p_jx_j^m\bigr)e^{iu\widehat{x}_j}
+p_jx_j^m\bigl(e^{iu\widehat{x}_j}-e^{iux_j}\bigr)\Bigr\} + \sum_{j \notin \widehat{\mJ}} i^m x_j^{m} p_je^{iux_j}.
\]
Since $\supp(\mu_d)$ is bounded, we get, similarly to~\eqref{bound_phi_d} and \eqref{615},
%$|x_j|,|\widehat x_j|\le C_x$ for some constant
%$C_x<\infty$, hence
%\[
%\bigl|\widehat p_j\widehat x_j^m-p_jx_j^m\bigr|\lesssim|\widehat p_j-p_j|+|\widehat x_j-x_j|,
%\qquad
%\bigl|e^{iu\widehat x_j}-e^{iux_j}\bigr|\le|u|\,|\widehat x_j-x_j|\le U_n|\widehat x_j-x_j|
%\]
%for $|u|\le U_n$. Similarly to Theorem \ref{thm3} summing over $j$ yields,
%\[
%\sup_{u \in [-V_n, V_n]}|\widehat\phi_d^{(m)}-\phi_d^{(m)}|
%\lesssim \max_j|\widehat p_j-p_j|+V_n\max_j|\widehat x_j-x_j| + \sum_{j \notin  \widehat\mJ}
%           \bigl( \widehat{p}_j + |\widehat{p}_j - p_j|
%           \bigr).
%\]
%By Theorem \ref{thm1} we have
%\begin{align*}
%    \max_j|\widehat p_j-p_j|=O_\P\Bigl(\mQ_n\Bigr), \quad
%    \sup_{\mathcal{S}} \max_{j \in \mJ \cap \hat{\mJ}}|\widehat x_j-x_j|=O_\P\Bigl(\mQ_n\Bigr).
%\end{align*}
\begin{equation*}
    \sup_{u\in [-W_n, W_n]}\bigl|\widehat\phi_d^{(m)}(u)-\phi_d^{(m)}(u)\bigr|
=O_\P\Bigl(W_n \mQ_n\Bigr)
\end{equation*}
and the desired result follows. Similarly to the approach in Theorem \ref{thm3}, we have
\[
\sup_{u\in [-W_n, W_n]}\bigl|\widehat\psi_d''(u)-\psi_d''(u)\bigr|
=O_\P\Bigl(W_n\mQ_n\Bigr).
\]
Since $\supp\F[K_n]\subseteq[-W_n,W_n]$, the $H^{-1}$-norm of $I_1^d$ is equal to 
\begin{align*}
\|I_1^d\|^2_{H^{-1}}
&=\frac1{2\pi}\int_{-W_n}^{W_n}\frac{|\widehat\psi_d''-\psi_d''|^2\,|\F[K_n](u)|^2}{1+u^2}\,du \\
&\leq\sup_{u\in [-W_n, W_n]}|\widehat\psi_d''-\psi_d''|^2\cdot\int_{-1}^{1}(1+W_n^2v^2)^{-1}|\F[K](v)|^2W_n\,dv =O_\P\Bigl( W_n^{2} \mQ_n^2\Bigr).
\end{align*}
Hence, $\|I_1^d\|_{H^{-1}}=O_\P\Bigl(W_n\mQ_n\Bigr)$. Combining the bounds for $I_1^d$ and $I_2^d$ we obtain \begin{align*}\|\widehat{\bar\nu}_d-\bar\nu_d\|_{H^{-1}}=O_\P\bigl(W_n\mQ_n+C_2 W_n^{-1/2}\bigr),\end{align*} which concludes the proof.\end{proof}

\section*{Acknowledgments}
This article is an output of a research project HSE-BR-2025-039 implemented as part of the Basic Research Program at HSE University.

%\section*{Funding}
%Funding information here.

%\bibliographystyle{plain}
%\bibliography{reference}

\bibliographystyle{plainnat}
\bibliography{references}

%USE THE BELOW OPTIONS IN CASE YOU NEED AUTHOR YEAR FORMAT.
%\bibliographystyle{abbrvnat}
%\bibliography{reference}

\end{document}